\documentclass[12pt]{article}
\usepackage[T1]{fontenc}
\usepackage{kpfonts}
\usepackage[dvipsnames]{xcolor}
\usepackage{pdfpages}
\usepackage{sgame}
\usepackage{accents}
\usepackage{tikz-cd}
\usepackage{float}
\usepackage[round]{natbib}
\usepackage{multirow}
\usepackage{dutchcal}
\usepackage{pdfpages}
\usepackage{bbm}
\usepackage{sgame}
\usepackage{mathtools}
\usepackage{amsthm}
\usepackage{enumitem}
\usepackage[margin=1in]{geometry}
\usepackage{caption}
\usepackage{subcaption}
\usepackage{epigraph}
\usepackage{listings}
\usetikzlibrary{calc}
\usetikzlibrary{shapes,arrows}
\usepackage{tcolorbox}

\DeclareMathOperator\supp{supp}
\DeclareMathOperator\inter{int}
\DeclareMathOperator\conv{conv}

\DeclareMathOperator\epi{epi}

\newtheorem{theorem}{Theorem}[section]
\newtheorem{proposition}[theorem]{Proposition}
\newtheorem{lemma}[theorem]{Lemma}
\newtheorem{corollary}[theorem]{Corollary}
\newtheorem{claim}[theorem]{Claim}

\theoremstyle{definition}

\newtheorem{scenario}{Scenario}
\newtheorem{definition}[theorem]{Definition}

\newtheorem{remark}[theorem]{Remark}

\usepackage{hyperref}
\hypersetup{
    colorlinks=true,
    linkcolor=OrangeRed,
    filecolor=Periwinkle,      
    urlcolor=Periwinkle,
    citecolor=Periwinkle,
}

\definecolor{backcolour}{rgb}{0.63, 0.79, 0.95}

\lstdefinestyle{mystyle}{
  backgroundcolor=\color{backcolour},
  basicstyle=\ttfamily\footnotesize,
  breakatwhitespace=false,         
  breaklines=true,                 
  captionpos=b,                    
  keepspaces=true,                 
  numbers=left,                    
  numbersep=5pt,                  
  showspaces=false,                
  showstringspaces=false,
  showtabs=false,                  
  tabsize=2
}

\providecommand{\keywords}[1]{\textbf{\textit{Keywords:}} #1}
\providecommand{\jel}[1]{\textbf{\textit{JEL Classifications:}} #1}

\begin{document}
\title{Making Information More Valuable}
\author{ Mark Whitmeyer\thanks{Arizona State University. Email: \href{mailto:mark.whitmeyer@gmail.com}{mark.whitmeyer@gmail.com}. Dedicated to Kylee Sutton. I am grateful to Simon Board, Costas Cavounidis, Gregorio Curello, Dana Foarta, Rosemary Hopcroft, Vasudha Jain, Doron Ravid, Eddie Schlee, Ludvig Sinander, Bruno Strulovici, Can Urgun, Tong Wang, Joseph Whitmeyer, Tom Wiseman, Renkun Yang, Kun Zhang, and various seminar and conference audiences for their feedback. I also thank the editor, Emir Kamenica, and four anonymous referees for their useful suggestions. This paper was formerly titled ``Flexibility and Information.''}}

\date{\today}

\maketitle

\begin{abstract}
We study what changes to an agent's decision problem increase her value for information. We prove that information becomes more valuable if and only if the agent's reduced-form payoff in her belief becomes more convex. When the transformation corresponds to the addition of an action, the requisite increase in convexity occurs if and only if a simple geometric condition holds, which extends in a natural way to the addition of multiple actions. We apply these findings to two scenarios: a monopolistic screening problem in which the good is information and delegation with information acquisition.
\end{abstract}
\\
\jel{D81; D82; D83}

\newpage

\setlength{\epigraphwidth}{3.75in}\begin{epigraphs}
\qitem{When action grows unprofitable, gather information; when information grows unprofitable, sleep.}%
      {Ursula Le Guin, \textit{The Left Hand of Darkness}}
\end{epigraphs}

\section{Introduction}

To a rational, expected-utility maximizing decision-maker, information is always valuable. However, some kinds of information are more valuable than others. Likewise, information is more valuable in certain decision problems than in others. The economics literature, commencing with the seminal work of Blackwell (\cite{blackwell} and \cite{blackwell2}), has largely focused on the first comparison, between information structures. In this paper, we study the second comparison, between decision problems. When can we say that one agent values information more than another? Equivalently, suppose we alter an agent's decision problem. What sorts of modifications increase her value for information?

Blackwell's way of comparing information structures is relatively detail-free. In his ranking, information structure \(1\) is more valuable than information structure \(2\) if for \textit{any} prior held by the agent and \textit{any} decision problem, the agent prefers \(1\) to \(2\). In this paper, in specifying what it means for agent \(1\) to value information more than agent \(2\), we take a similarly broad approach. Namely, we require that \textit{any} information structure be more valuable to agent \(1\) than agent \(2\), \textit{no matter the prior}.

We begin by identifying that relative \textbf{convexity} distinguishes agents' comparative love of information. Agent \(1\) values information more than agent \(2\) if the difference in the agents' value functions\footnote{An agent's value function, \(V\left(\mu\right)\), is her maximal expected payoff at any belief \(\mu \in \Delta\left(\Theta\right)\), obtained by plugging in an optimizing action. We define this object formally on page \(6\) in Expression \ref{valuefunction}.} \(V_1-V_2\) is convex. Next, we turn our attention to the value functions themselves. What modifications to an agent's decision problem result in an increase in convexity? One natural way to alter an agent's decision problem is by allowing her an additional action. How does increased flexibility--a greater capacity to adapt her behavior to new information--change an agent's value for information?

In \(\S\)\ref{secmuch}, we carry this analysis further by allowing the agent not just one but potentially multiple additional actions. Next, in \(\S\)\ref{lessflexsec}, we remove actions. We then leave the set of actions unchanged (\(\S\)\ref{transsec}), but instead scale the agent's utility. This allows us to speak to the effects of repetition and aggregate risk on the value of information. In \(\S\)\ref{riskaverse}, we reveal that increased (or decreased) risk aversion has an ambiguous affect on an agent's value for information. 

Central to our study is the observation that a modification to an agent's decision problem alters her value for information through two channels. The first is the agent's sensitivity to information--if her value for information increases (in the manner defined in this paper) it must be that she is more reactive to information. The second is the value to the agent of distinguishing between actions. This must also increase if the agent's value for information is to increase. All in all, a transformation makes information more valuable if and only if the agent becomes more sensitive to information and the value of distinguishing between actions increases.

When the transformation to the decision problem is either the addition or subtraction of actions, the first channel is all-important. This is particularly stark when we modify the agent's decision problem by adding a single action (\(\S\)\ref{secmore}): the agent's value for information increases if and only if she becomes more sensitive to information. We uncover a simple geometric condition necessary and sufficient for this to transpire and show that an iterative version of this condition also guarantees an increase in the value of information when multiple actions are added. Moreover, although this condition is not necessary to make information more valuable when multiple actions are added, any failure of necessity is not robust--perturbing the utilities from the new actions slightly will make it so that the agent does not become more sensitive to information.

Perhaps unexpectedly, we discover that unless all of the remaining actions or all of the removed actions were initially dominated, taking away actions can never lead to a higher value for information. That is, it is only an elimination that results in a totally new decision problem (in effect) or the exact same decision problem, that can lead to an increase in an agent's value for information. Any removal other than these necessarily makes the agent less sensitive to information.



\subsection{Motivating Example}

The question under study has significant practical relevance. The job of a regulator is to enact policies that modify the incentives of agents in some environment. This typically entails the addition or subtraction of actions: there are contracts that an insurer may not offer, assets that an investment firm may not sell, and limits to how many fish a trawler may catch.\footnote{The verdicts that can be handed down in criminal cases are also legislated, so our results also speak to what kinds of verdicts improve incentives for information acquisition (cf. \cite{siegel2020economic}).} Insurers themselves change agents' payoffs by reducing their risk, flattening their payoffs. Firms do the opposite with their workers: bonus schemes tied to a worker's performance make her payoff steeper and more sensitive to randomness.

Consider for instance an insurance provider dictating what treatments it will cover; \textit{viz.,} what procedures a doctor may conduct. For simplicity, suppose there are three conditions a patient with an injured hand may have--three states of the world. In one state, state \(0\), the injury is just a sprain; in another, state \(1\), a bone is broken but not displaced; and in state \(2\), the fracture is displaced.

Suppose first the doctor may only offer one treatment: place a cast on the hand (action \(c\)). Accordingly, she has two possible actions, do nothing (action \(n\)), which is uniquely optimal if the injury is just a sprain; or cast the hand, which is uniquely optimal if the hand is broken. This decision problem is represented in Figure \ref{exsub1}: point \(\left(\mu_1,\mu_2\right)\) specifies the respective probabilities (beliefs) that the bone is broken but not displaced or broken and displaced. Accordingly, the plain blue region is the region of probabilities in which \(n\) is optimal; and the patterned red region are those probabilities for which \(c\) is optimal.

Let us now consider two possible new treatments afforded to the doctor. Suppose the provider now covers surgery (action \(s\)). This is relatively high-risk and is only optimal if the doctor is confident the bone is broken and displaced. Figure \ref{exsub2} represents this new decision problem: \(s\) is optimal if and only if the doctor's belief is in the purple triangle on the top left. On the other hand, suppose the provider instead allows a conservative treatment consisting of stretching and rehabilitating exercises (action \(r\)). This is better than nothing in the case of a fracture, but is inferior to rest for sprains. This scenario is Figure \ref{exsub3}, where \(r\) is optimal for beliefs in the central dark gray region.

Which of these new options, if either, does not dampen the doctor's enthusiasm for information, regardless of her prior or what that information may be? As we discover in this paper, the answer is simple, only the former of the two potential new procedures, surgery, makes information more valuable. Indeed, suppose that a sprain and a break are equally likely. With only the initial two treatments to choose from, the doctor strictly benefits from any information. If we gave the doctor the conservative option, this would clearly no longer be true: any information that doesn't move her beliefs much is now worthless, as the conservative treatment remains optimal at those beliefs. In contrast, the surgery option makes information weakly more valuable. 

The crucial difference between the two potential new actions is that the surgery option is \textbf{refining}: only the region of beliefs in which doing nothing is optimal shrinks. In contrast, the conservative treatment partially replaces each pre-existing treatment.

\begin{figure}
\centering
\begin{subfigure}{.5\textwidth}
  \centering
  \includegraphics[scale=.22]{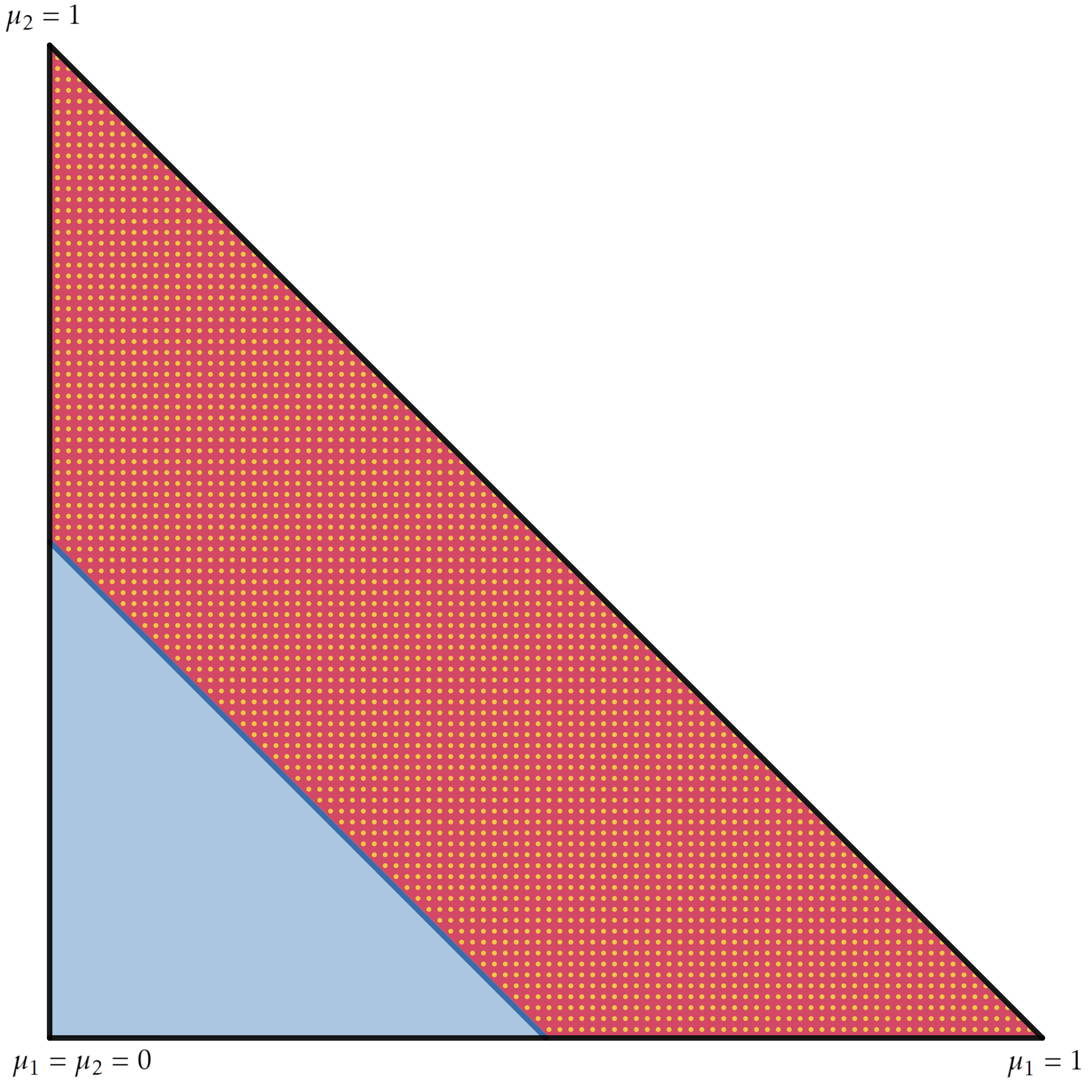}
  \caption{Choosing between \(c\) and \(n\)}
  \label{exsub1}
\end{subfigure}
\par
\bigskip
\par
\bigskip
\par
\begin{subfigure}{.5\textwidth}
  \centering
  \includegraphics[scale=.22]{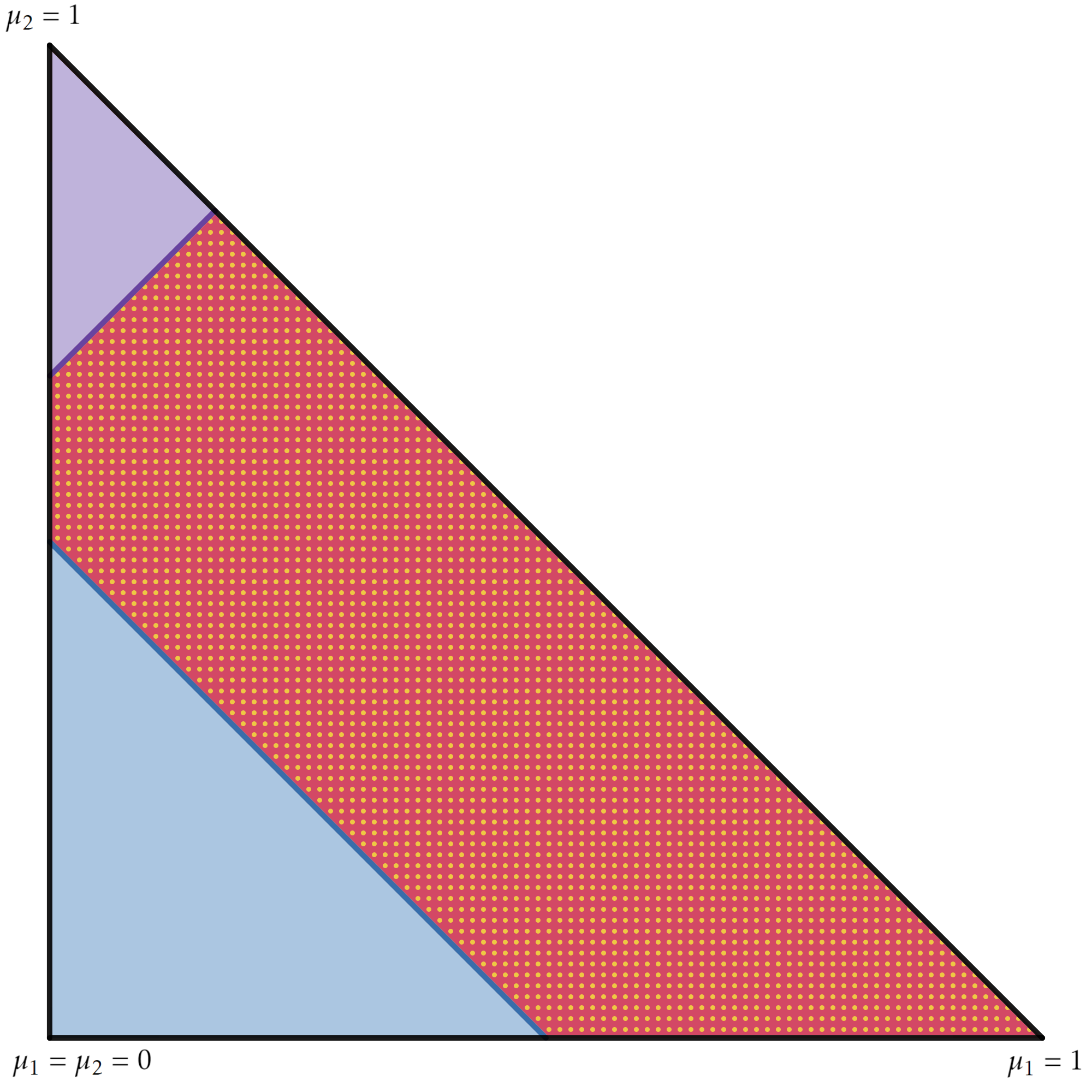}
  \caption{Choosing between \(c\), \(n\), and \(s
  \)}
  \label{exsub2}
\end{subfigure}%
\begin{subfigure}{.5\textwidth}
  \centering
  \includegraphics[scale=.22]{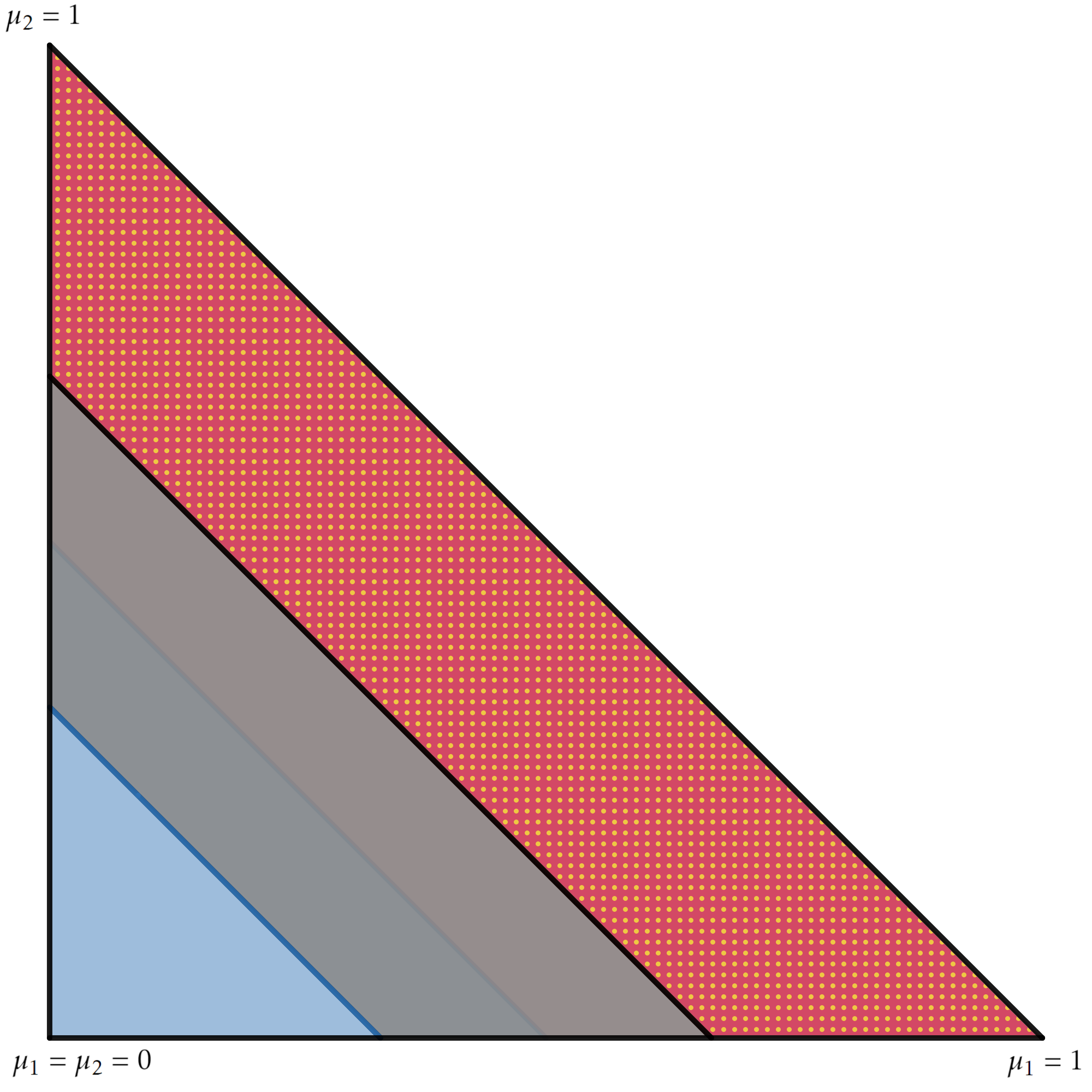}
  \caption{Choosing between \(c\), \(n\), and \(r\)}
  \label{exsub3}
\end{subfigure}
\caption{The subdivision of the \(2\)-simplex corresponding to the motivating example. There are three states \(\left\{0,1,2\right\}\), with \(\mu_1 \coloneqq \mathbb{P}\left(1\right)\) and \(\mu_2 \coloneqq \mathbb{P}\left(2\right)\).}
\label{exex}
\end{figure}

\section{The Model}

There is a grand set of actions \(\mathcal{A}\). Our protagonist is a decision-maker, an agent who initially possesses a compact set of actions \(A \subseteq \mathcal{A}\). There is an unknown state of the world \(\theta\), which is drawn according to some full-support prior \(\mu_{0}\) from some finite set of states \(\Theta\). Initially, the agent has some continuous utility function \(u \colon \mathcal{A} \times \Theta \to \mathbb{R}\).

\(\mathcal{D} \coloneqq \left(u,A,\Theta\right)\) denotes the agent's \emph{Initial Decision Problem}. We are studying the effect of a transformation of the decision problem on the agent's value for information.\footnote{Equivalently, we are comparing two different agents' values for information.} To that end, \(\hat{\mathcal{D}} \coloneqq \left(\hat{u},\hat{A},\Theta\right)\) denotes the agent's \emph{Transformed Decision Problem}. Here are a few leading examples of such transformations:

\begin{scenario}
    \textbf{Becoming More Flexible.} \(A\) is finite. In \(\hat{\mathcal{D}}\), the agent's utility function remains unchanged, \(\hat{u} = u\); and her new set of actions is \(\hat{A} \coloneqq A \cup  B\) for some additional finite set of actions \(B\in \mathcal{A} \setminus A\).
\end{scenario}

\begin{scenario}
    \textbf{Becoming A Little More Flexible.} This is the special case of the agent becoming more flexible in which the set of actions being added, \(B\), is a singleton. That, is \(B = \left\{\hat{a}\right\}\) for some \(\hat{a} \in \mathcal{A} \setminus A\).
\end{scenario}

\begin{scenario}
    \textbf{Becoming Less Flexible.} Yet again, \(A\) is finite and \(\hat{u} = u\). Now, however, the agent loses actions moving from \(\mathcal{D}\) to \(\hat{\mathcal{D}}\): \(\emptyset \neq \hat{A} \subset A\).
\end{scenario}

\begin{scenario}
    \textbf{Transforming the Agent's Utility Function.} In \(\hat{\mathcal{D}}\), \(\hat{u} = \phi \circ u\) for some strictly increasing, continuous \(\phi\); and her new set of actions is unaltered: \(\hat{A} = A\).
\end{scenario}

We say that \(\mathcal{D}\) (\(\hat{\mathcal{D}}\)) is \emph{Finite} if \(A\) (\(\hat{A}\)) is finite. \(\Delta \left(\Theta\right)\) denotes the simplex of probability distributions over \(\Theta\) and \(\inter \Delta \left(\Theta\right)\) its interior.\footnote{For a set \(Y\), \(\inter Y\) denotes its relative interior.} When the agent acquires information, she does so by observing the realization of an experiment, stochastic map \(\pi \colon \Theta \to \Delta\left(S\right)\), where \(S\) is a compact set of signal realizations. Equivalently (\cite*{kam}), she obtains a Bayes-plausible distribution over posterior beliefs (posteriors) \(\Phi \in \mathcal{F}\left(\mu_{0}\right) \subseteq \Delta \left(\Delta\left(\Theta\right)\right)\).\footnote{A distribution \(\Phi\) is Bayes-plausible if it is supported on a subset of \(\Delta \left(\Theta\right)\) and \(\mathbb{E}_{\Phi}\left(\mu\right) = \mu_{0}\).}

Given an initial decision problem \(\mathcal{D}\), the agent's \emph{Value Function}, in belief \(\mu \in \Delta\left(\Theta\right)\), is
\[\label{valuefunction}\tag{\(1\)} V\left(\mu\right) \coloneqq \max_{a \in A}
\mathbb{E}_{\mu}u\left(a,\theta\right) \text{.}\]
\(\hat{V}\) is the analogous object in the transformed decision problem, \(\hat{\mathcal{D}}\), and both functions are convex. There are two natural ways to understand an increase in an agent's value for information.

\begin{definition}[Exogenous Information]
    Given \(\mathcal{D}\) and \(\hat{\mathcal{D}}\), we say that \emph{The Transformation Generates a Greater Value for Information} if \[\mathbb{E}_{\Phi}\hat{V}\left(\mu\right) - \hat{V}\left(\mu_{0}\right) \geq \mathbb{E}_{\Phi}V\left(\mu\right) - V\left(\mu_{0}\right)\text{,}\] for all \(\Phi \in \mathcal{F}\left(\mu_{0}\right)\) and \(\mu_{0} \in \inter \Delta\left(\Theta\right)\).
\end{definition}
The transformation generates a greater value for information if the resulting expected payoff as a result of obtaining information is greater in the transformed decision problem than in the initial decision problem, no matter the information (for any Bayes-plausible \(\Phi\)) and no matter the prior (any \(\mu_{0} \in \inter \Delta \left(\Theta\right)\)).

Our second way of interpreting the value of information gives the agent greater control over information acquisition. Now, the experiment is not exogenous but instead an endogenous choice of the agent. Given an initial decision problem \(\mathcal{D}\) and a prior \(\mu_{0}\), in the agent's flexible information acquisition problem, she solves 
 \[\label{flex}\tag{\(2\)} \max_{\Phi \in \mathcal{F}\left(\mu_{0}\right)}\int_{\Delta\left(\Theta\right)}V\left(\mu\right)d\Phi\left(\mu\right) - D\left(\Phi\right)\text{,}\]
where \(D\) is a uniformly posterior-separable cost functional,\footnote{This family of costs is introduced in \cite*{caplin2022rationally}.} that for which there exists a strictly convex function \(c \colon \Delta\left(\Theta\right) \to \mathbb{R}\) such that
\[D\left(\Phi\right) = \int_{\Delta\left(\Theta\right)}c\left(\mu\right)d\Phi\left(\mu\right) - c\left(\mu_{0}\right) \text{.}\]
Similarly, in the transformed decision problem \(\hat{\mathcal{D}}\), the agent solves
 \[\label{flexhat}\tag{\(3\)} \max_{\Phi \in \mathcal{F}\left(\mu_{0}\right)}\int_{\Delta\left(\Theta\right)}\hat{V}\left(\mu\right)d\Phi\left(\mu\right) - D\left(\Phi\right)\text{,}\]
\begin{definition}[Endogenous Information]
    Given \(\mathcal{D}\) and \(\hat{\mathcal{D}}\), we say that \emph{The Transformation Does Not Generate Less Information Acquisition} if for any prior \(\mu_{0} \in \inter{\Delta\left(\Theta\right)}\), UPS cost functional \(D\), and solution to the agent's information acquisition problem in the initial decision problem (Problem \ref{flex}), \(\Phi^{*}\), there exists a solution to the agent's information acquisition problem in the transformed decision problem (Problem \ref{flexhat}), \(\hat{\Phi}^{*}\), that is not a strict mean-preserving contraction (MPC) of \(\Phi^{*}\).\footnote{For distributions \(P\) and \(Q\) supported on a compact, convex subset, \(X\), of a vector space, \(P\) is an MPC of \(Q\) if \(\int \phi dP \leq \int \phi dQ\) for all convex functions \(\phi \colon X \to \mathbb{R}\). \(P\) is a strict MPC of \(Q\) if \(P\) is an MPC of \(Q\) but \(Q\) is not an MPC of \(P\). \(Q\) is a mean-preserving spread (MPS) of \(P\) if \(P\) is an MPC of \(Q\).} 
\end{definition}
 
The transformation does not generate less information acquisition if for any optimal information acquisition strategy in \(\mathcal{D}\), there is an optimal information acquisition strategy in \(\hat{\mathcal{D}}\) in which the agent does not acquire strictly less information.

\section{Making Information More Valuable}\label{section2}

In this section, we begin by connecting the increased convexity of value function \(\hat{V}\) vis-a-vis \(V\) to an increased value for information. The remainder of the paper then studies what changes to the primitives of the decision problem lead to the requisite increased convexity. The central organizing result of the paper is the following theorem, which is at best a modest contribution in its own right.\footnote{As we shortly note, others have observed nearly-identical versions of the sufficiency results. Moreover, the necessity direction is closely related to classical results from the comparative-statics literature (e.g., Theorem 10 in \cite{monotonecomp}), once one realizes that ``adding a convex function is precisely what an `increasing-differences' shift means in this setting.'' I thank an anonymous referee for the quoted observation.}
\begin{theorem}\label{moreconvex}
    Given \(\mathcal{D}\) and \(\hat{\mathcal{D}}\), the following are equivalent:
    \begin{enumerate}[label={(\roman*)},noitemsep,topsep=0pt]
        \item \(\hat{V}-V\) is convex.
        \item The transformation does not generate less information acquisition.
        \item The transformation generates a greater value for information.
    \end{enumerate}
\end{theorem}

\medskip

First, we establish that the convexity of \(\hat{V}-V\) implies that a transformation generates a greater value for information. This is an implication of a stronger result that is a direct consequence of Blackwell's theorem. For a distribution \(\Phi\), \(MPC\left(\Phi\right)\) denotes the set of mean-preserving contractions (MPCs) of \(\Phi\).
\begin{lemma}\label{lemmany}
    If \(\hat{V} - V\) is convex, for any prior \(\mu_{0} \in \inter \Delta\left(\Theta\right)\) and distributions over posteriors \(\Phi, \Upsilon \in \mathcal{F}_{\mu_{0}}\) with \(\Upsilon \in MPC\left(\Phi\right)\),
    \[\mathbb{E}_{\Phi}\hat{V}\left(\mu\right) - \mathbb{E}_{\Upsilon}\hat{V}\left(\mu\right) \geq  \mathbb{E}_{\Phi}V\left(\mu\right) - \mathbb{E}_{\Upsilon}V\left(\mu\right)\text{.}\]
\end{lemma}
\begin{proof}
Suppose for the sake of contradiction that \[\mathbb{E}_{\Phi}\hat{V}\left(\mu\right) - \mathbb{E}_{\Upsilon}\hat{V}\left(\mu\right) <  \mathbb{E}_{\Phi}V\left(\mu\right) - \mathbb{E}_{\Upsilon}V\left(\mu\right)\text{,}\]
which holds if and only if
\[\mathbb{E}_{\Phi}\left[\hat{V}\left(\mu\right) - V\left(\mu\right)\right] < \mathbb{E}_{\Upsilon}\left[\hat{V}\left(\mu\right) - V\left(\mu\right)\right]\text{,}\]
which violates the definition of an MPC. \end{proof}
As no information corresponds to the degenerate distribution on the prior, \(\delta_{\mu_{0}}\), this lemma produces the desired implication, that \(\hat{V}-V\) being convex implies that a transformation increases an agent's value for information.

Second, we establish the sufficiency of \(\left(\hat{V}-V\right)\)'s convexity in the endogenous information case.
\begin{lemma}\label{endog}
Given \(\mathcal{D}\) and \(\hat{\mathcal{D}}\), if \(\hat{V} - V\) is convex, the transformation does not generate less information acquisition.
\end{lemma}
\begin{proof}
Fix an arbitrary prior \(\mu_{0} \in \inter \Delta\left(\Theta\right)\) and UPS cost functional \(D\). Let \(\Phi^*\) be an arbitrary solution to Problem \ref{flex}. Suppose for the sake of contradiction that in the transformed decision problem \(\hat{\mathcal{D}}\), every optimizer of Problem \ref{flexhat} is a strict MPC of \(\Phi\). Pick one, \(\hat{\Phi}^*\). The optimality of \(\hat{\Phi}^*\) and strict suboptimality of \(\Phi^*\) in \(\hat{\mathcal{D}}\) imply
\[\mathbb{E}_{\hat{\Phi}^*}\hat{V} - D\left(\hat{\Phi}^*\right) > \mathbb{E}_{\Phi^*}\hat{V} - D\left(\Phi^*\right)\text{.}\]
Analogously, the optimality of \(\Phi^*\) in \(\mathcal{D}\) implies
\[\mathbb{E}_{\Phi^*}V - D\left(\Phi^*\right) \geq \mathbb{E}_{\hat{\Phi}^*}V - D\left(\hat{\Phi}^*\right)\text{.}\]
Combining these two inequalities produces
\[\mathbb{E}_{\hat{\Phi}^*}\left[\hat{V} - V\right] > \mathbb{E}_{\Phi^*}\left[\hat{V} - V\right]\text{,}\]
contradicting that \(\hat{\Phi}^*\) is an MPC of \(\Phi^*\). \end{proof}
Note that in Lemma \ref{endog} we make no use of the fact that the cost functional is UPS; indeed, we do not even make use of the fact that it is monotone in the Blackwell order. All that is required is that the agent's payoff is additively separable in her value from the decision problem and her cost of acquiring information. Posterior separability of the cost function, instead, disciplines the necessity portion of Theorem \ref{moreconvex}, \textit{infra}. After all, our argument by contraposition would only be made easier by allowing for more general cost functionals (we already know a UPS one will do the trick).

This lemma is not new, though the proof is. The result first appears as Proposition 2 in \cite*{chambers2020costly}; and is, moreover, a corollary of stronger results in \cite{yoder2022designing} and \cite{denti2022posterior}.
\cite{denti2022posterior} and \cite{chambers2020costly} are similar in aims: both seek to understand which datasets are consistent with costly information acquisition. In his discussion of external validity, \cite{denti2022posterior} writes ``our analysis points to a general property of posterior separable costs that may inform the answer to this question: increasing the incentive to acquire information leads to more extreme beliefs.'' His ensuing proposition (3) states that if \(\hat{V} - V\) is convex, and strictly convex at the prior, then for any solution to the agent's information acquisition problem in the initial decision problem (Problem \ref{flex}), \(\Phi^*\), and any solution to the agent's information acquisition problem in the transformed decision problem (Problem \ref{flexhat}), \(\hat{\Phi}^*\), the support of \(\hat{\Phi}^*\) cannot lie in the relative interior of the convex hull of the support of \(\Phi^*\).

\cite{chambers2020costly} focus on the implications of additive separability in models of costly information acquisition. To that end, they note that the additively-separable model ``forbids an individual from choosing a less informative information structure when there are `higher gross return from information;' '' i.e., \(\hat{V}-V\) is convex.\footnote{Interestingly, both \cite{chambers2020costly} and \cite{denti2022posterior} observe the equivalence of \(\left(\hat{V}-V\right)\)'s convexity with a transformation generating a greater value for information, but do not give proofs.} In his screening model, in order to characterize the responses to menus of contracts by agents with varying abilities to acquire information, \cite{yoder2022designing} establishes a stronger version of Lemma \ref{lemmany}. His Proposition 4 states that \(\left(\hat{V}-V\right)\)'s convexity implies that the intersection of the support of any \(\hat{\Phi}^*\) with the convex hull of the support of any \(\Phi^*\) is a (possibly empty) subset of the extreme points of the convex hull of the support of \(\Phi^*\).

Third, we turn our attention to necessity. It is easiest to start with the endogenous information case.
\begin{lemma}\label{newendoginfo}
    Given \(\mathcal{D}\) and \(\hat{\mathcal{D}}\), if the transformation does not generate less information acquisition, \(\hat{V} - V\) is convex.
\end{lemma}
\begin{proof}
    Please visit Appendix \ref{newendoginfoproof}.
\end{proof}
We prove this result by contraposition. If \(\hat{V}-V\) is convex, we can construct a cost function that is such that no matter her prior, an agent with value function \(\hat{V}\) strictly prefers to acquire no information. In contrast, there are some priors at which an agent with value function \(V\) strictly prefers to acquire some information, yielding the result.

Fourth, necessity in the exogenous information case is an easy consequence of the endogenous information lemma. Indeed, we may just take the specific distributions generated in the previous lemma's proof ``off-the-shelf.''
\begin{lemma}\label{newexoginfo}
Given \(\mathcal{D}\) and \(\hat{\mathcal{D}}\), if the transformation generates a greater value for information, \(\hat{V} - V\) is convex.
\end{lemma}
\begin{proof}
Contained in Appendix \ref{newexoginfoproof}.\end{proof}

If there are just two states, a stronger statement holds concerning an agent's optimal information acquisition as a result of a transformation. Given an initial decision problem \(\mathcal{D}\) and a transformed decision problem \(\hat{\mathcal{D}}\), the transformation \emph{Generates More Information Acquisition} if for any prior \(\mu_{0} \in \inter{\Delta\left(\Theta\right)}\), UPS cost functional \(D\), and solution to the agent's information acquisition problem in the initial decision problem (Problem \ref{flex}), \(\Phi^{*}\), there exists a solution to the agent's information acquisition problem in the transformed decision problem (Problem \ref{flexhat}), \(\hat{\Phi}^{*}\), that is a mean-preserving spread of \(\Phi^{*}\).
    
\begin{proposition}\label{twostates} Given \(\mathcal{D}\) and \(\hat{\mathcal{D}}\), if \(\left|\Theta\right| = 2\), \(\hat{V}-V\) is convex if and only if the transformation generates more information acquisition.
\end{proposition}
\begin{proof}
    Theorem \ref{moreconvex} implies the necessity portion of the result. Sufficiency is a consequence of the aforementioned Proposition 4 in \cite{yoder2022designing}.
\end{proof}

\subsection{``More Information'' Only in Trivial Cases For Three or More States}

When \(\left|\Theta\right| > 2\), the convexity of \(\hat{V}-V\) does not mean that the agent will acquire more information in the transformed decision problem. This is because the new and old experiments may not be (Blackwell) comparable. As we discuss above, \cite{yoder2022designing} and \cite{denti2022posterior} reveal that slightly stronger statements can be made about any solutions to Problems \ref{flex} and \ref{flexhat}, but these results are still weaker than saying the agent must acquire more information. In this subsection, we show that when there are three or more states, a transformation must lead to more information acquisition only in trivial cases.

\begin{figure}
\centering
\begin{subfigure}{.5\textwidth}
  \centering
  \includegraphics[scale=.35]{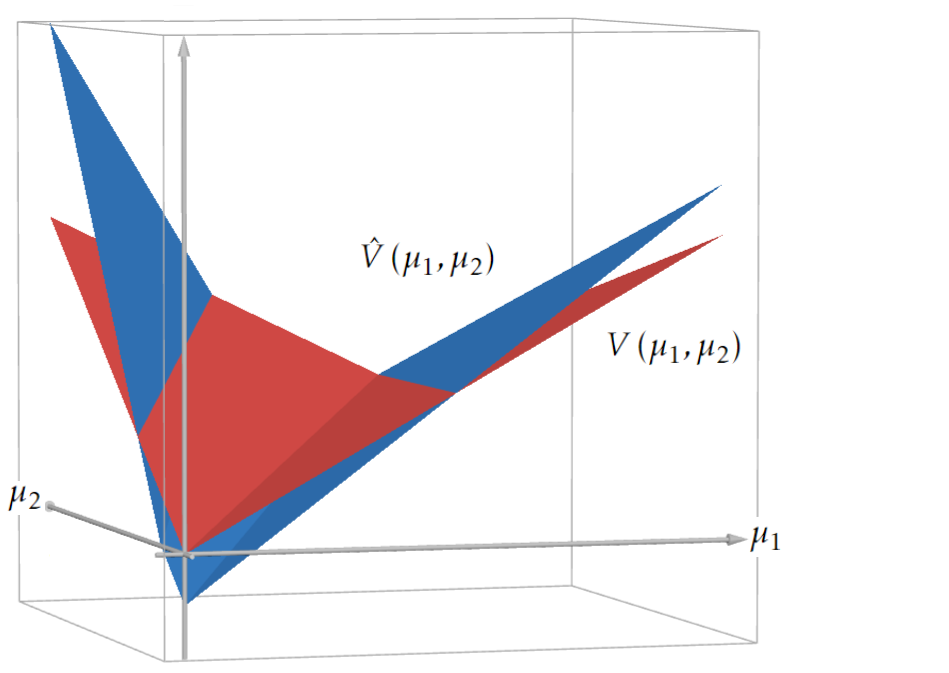}
  \caption{\(V\) (in light red) and \(\hat{V}\) (in dark blue).}
  \label{57figsub1}
\end{subfigure}%
\begin{subfigure}{.5\textwidth}
  \centering
  \includegraphics[scale=.45]{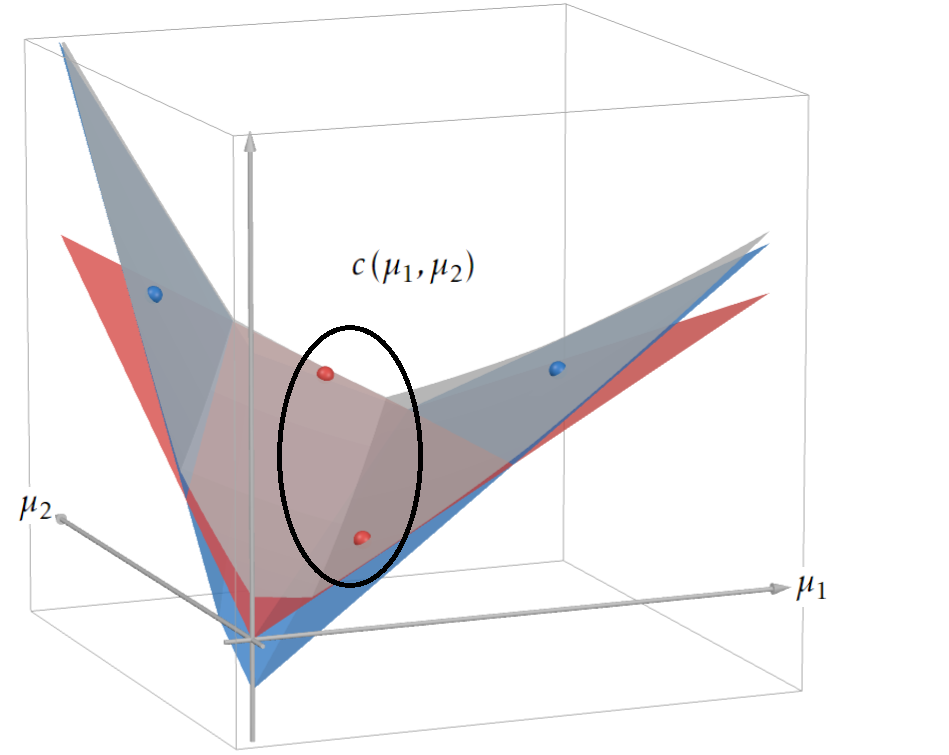}
  \caption{The desired cost function (in light grey).}
  \label{57figsub2}
\end{subfigure}
\caption{Proving Proposition \ref{nomoreprop}. The distributions supported on the points within the dark oval and the points not in the oval are unique solutions to Problems \ref{flex} and \ref{flexhat}, respectively, yet are Blackwell incomparable.}
\label{fig57}
\end{figure}
\begin{proposition}\label{nomoreprop}
Given finite \(\mathcal{D}\) and \(\hat{\mathcal{D}}\), if \(\left|\Theta\right| \geq 3\), the transformation generates more information acquisition if and only if \(\hat{V}-V\) or \(V\) is affine.  
\end{proposition}
\begin{proof}
Sufficiency is almost immediate. For necessity, there is a simple proof via contraposition, illustrated in Figure \ref{fig57}. As neither \(\hat{V}-V\) nor \(V\) is affine, and as decision problems are unaltered by the addition of affine functions, we specify without loss of generality that there is a region where \(V\) lies strictly above \(\hat{V}\) (and is not affine on that region), and that the two curves (\(\hat{V}\) and \(V\)) intersect in such a way that the region where \(\hat{V}\) lies strictly above \(V\) is not convex. This is illustrated in Figure \ref{57figsub1}. Note that we can still do this when there are just two states; in fact, this is one way to deliver a concise proof of Proposition \ref{twostates}.

Next, we use Lemma \ref{nonredundancy} to claim the existence of a cost function that intersects \(\max \left\{\hat{V},V\right\}\) at two pairs of points, the line segments between which form a ``cross.'' Furthermore, one of these pairs lies in the region where \(\hat{V}\) is strictly larger than \(V\) and the the other pair in the region where \(V\) is strictly larger than \(\hat{V}\). ``Cross'' refers to the fact that there is a unique intersection point between the line segments connecting each point in a pair; \textit{viz.,} a common point in the (interior of the) convex hull of each pair of points. Note that this also means that the four points do not lie on the same line. This construction is illustrated in Figure \ref{57figsub2}. This portion of the argument is where three or more states is essential. With just two states, all points are collinear.

By design, when the prior is the meeting point of the two beams of the ``cross'' and the value function is \(\hat{V}\), one of the pairs is uniquely optimal; and when the value function is \(V\), the other is uniquely optimal. As the supports of these binary distributions are not collinear, they are Blackwell incomparable. The details lie in Appendix \ref{nomorepropproof}.
\end{proof}

\subsection{The Geometry of the Finite-Action Problem}\label{geometry}

In this section we introduce a particular geometric object corresponding to a value function, a polyhedral subdivision, as a representation of a decision problem. This representation meaningfully captures the essentials of how transformations make information more valuable and, therefore, is useful [for the analyses] in subsequent sections.  To that end, in understanding what changes to a decision problem increase an agent's value for information, it is helpful to categorize the effects of an alteration according to the classic extensive- intensive-margin dichotomy from labor economics. On the one hand, if the transformation increases an agent's value for information, it must be that the agent's behavior is more responsive to information: if the same action is optimal at two beliefs that led to distinct pre-transformation actions, then the transformation cannot increase the value of information. We can think of an increase in responsiveness as a positive change along the extensive margin. On the other hand, it may be that the agent's behavior is more responsive to information, yet the value of information decreases. This occurs when the actions become more similar payoff-wise--the value of distinguishing between actions decreases. In a sense, such a decline is a negative change along the intensive margin. 

Now let us formalize the informal argument from the previous paragraph regarding the extensive-margin effects of a change to the agent's decision problem. We will show the necessity of increased sensitivity to information for the value of information to increase. In particular, when the agent has finitely many actions, we note that polyhedral subdivisions are exact representations of the agent's sensitivity to information. Then, there is a natural (partial) ordering of subdivisions that corresponds exactly to this reactivity.

When \(A\) is finite, we say an action \(a_i \in A\) is not \emph{Weakly Dominated}, or is \emph{Undominated}, if there exists a belief \(\mu\in \Delta\left(\Theta\right)\) such that \[\mathbb{E}_{\mu}u\left(a_i,\theta\right) > \max_{a \in A\setminus\left\{a_i\right\}}\mathbb{E}_{\mu}u\left(a,\theta\right)\text{.}\]
If \(A\) is finite, the agent's value function, \(V\), is piecewise affine, and its graph is a polyhedral surface in \(\mathbb{R}^{n}\), where \(n\) is the number of states. Associated with \(V\) is the projection of its epigraph onto \(\Delta\left(\Theta\right)\), which yields a finite collection \(C\) of polytopes of full dimension \(C_i\) (\(i = 1, \dots, m\)), where \(m\) is the number of undominated actions in \(A\). Formally, for each undominated \(a_i \in A\) (\(i = \left\{1,\dots, m\right)\)), 
\[C_i \coloneqq \left\{\mu\in \Delta\left(\Theta\right) \ \vert \ \mathbb{E}_\mu\left(a_i,\theta\right) = V\left(\mu\right)\right\} = \left\{\mu \in \Delta\left(\Theta\right) \ \left| \ \mathbb{E}_{\mu}u\left(a_i,\theta\right) \geq \max_{a \in A\setminus\left\{a_i\right\}}\right.\mathbb{E}_{\mu}u\left(a,\theta\right)\right\}\text{.}\]
By construction, action \(a_i\) is optimal for any belief \(\mu \in C_i\) and uniquely optimal for any belief \(\mu \in \inter C_i\). The collection \(C\) is a \emph{Regular Polyhedral Subdivision} (henceforth, \emph{Subdivision}) of \(\Delta\left(\Theta\right)\). Figure \ref{fig1} illustrates two pairs of value functions and subdivisions when the agent has three actions. Each \(C_i\) is a \emph{Cell} of \(C\).  

\begin{figure}
\centering
\begin{subfigure}{.5\textwidth}
  \centering
  \includegraphics[scale=.18]{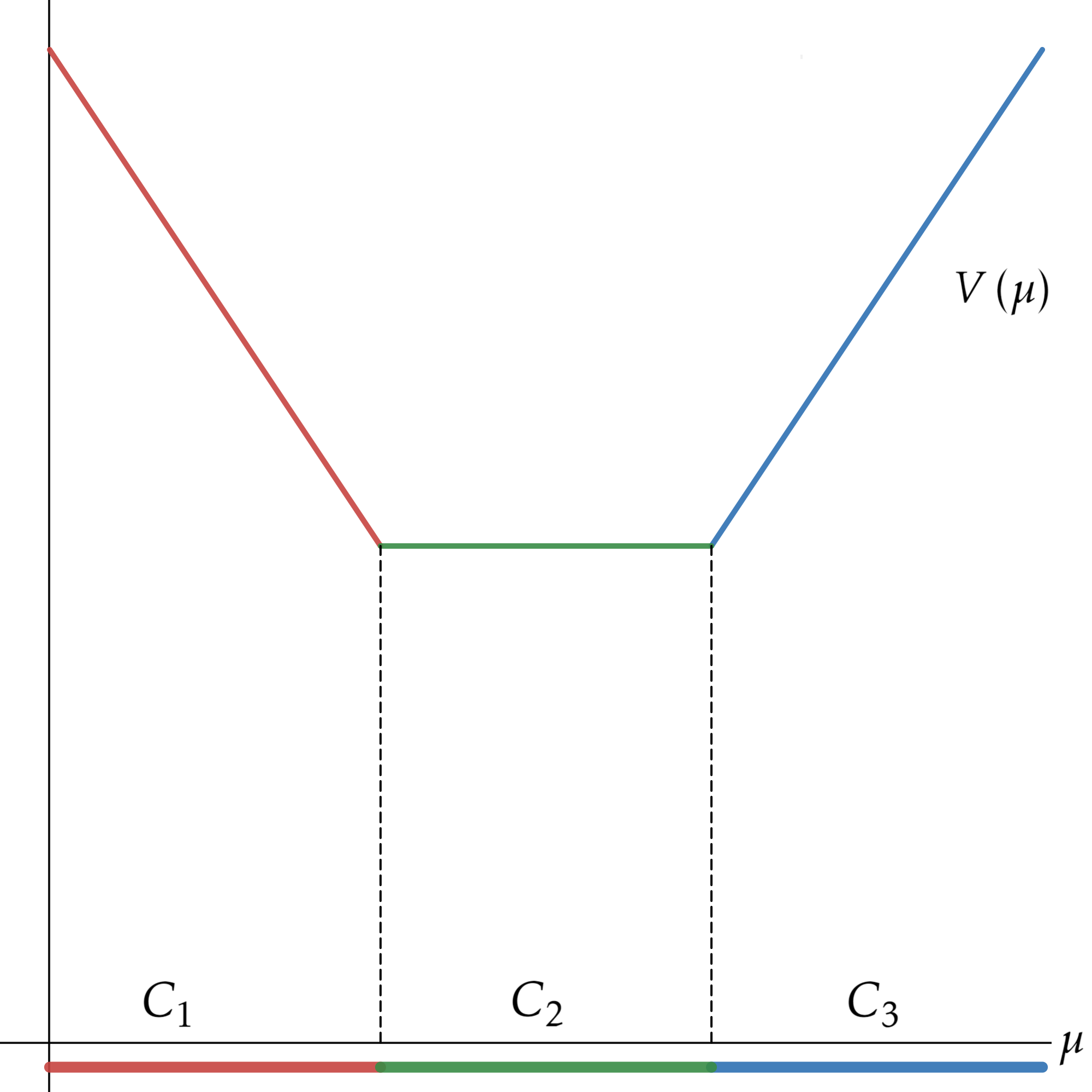}
  \caption{\(V\) and \(C\), two states.}
  \label{figsub12}
\end{subfigure}%
\begin{subfigure}{.5\textwidth}
  \centering
  \includegraphics[scale=.35]{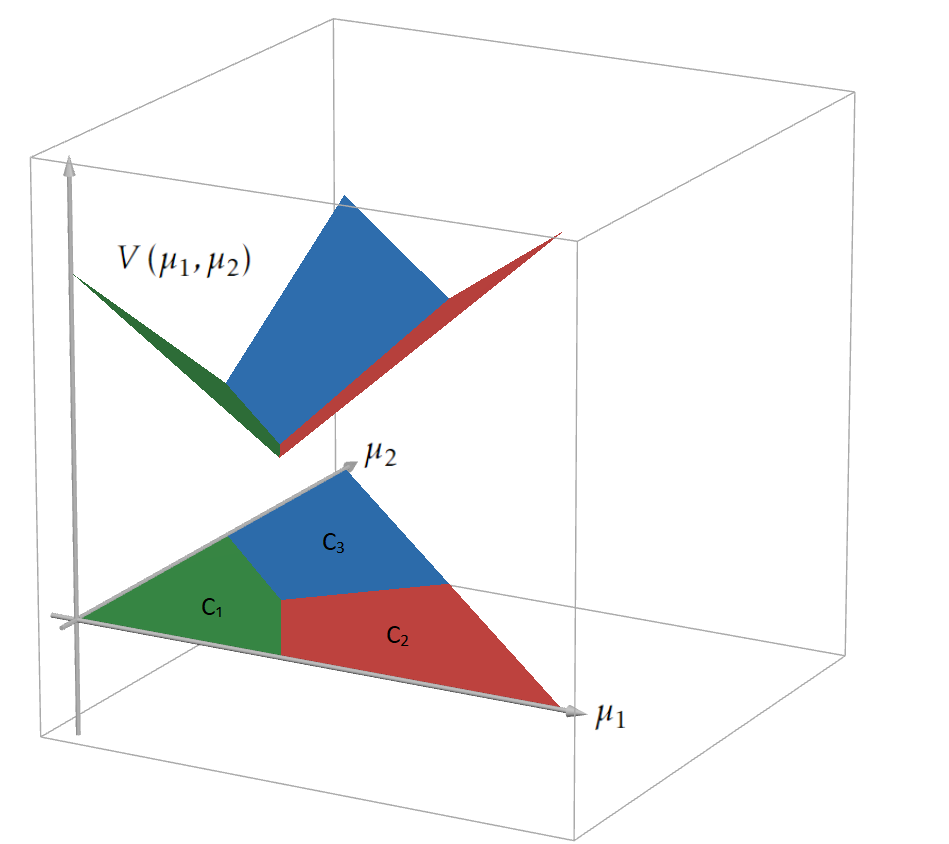}
  \caption{\(V\) and \(C\), three states.}
  \label{figsub22}
\end{subfigure}
\caption{}
\label{fig1}
\end{figure}

If the set of actions in the transformed decision problem, \(\hat{A}\), is also finite, the new value function \(\hat{V}\), itself, has a corresponding subdivision, \(\hat{C}\). There is a natural way of comparing subdivisions that is useful for our purposes: a subdivision \(P = \left\{P_1, \dots, P_l\right\}\) is \emph{Finer} than (or \emph{Refines}) a subdivision \(Q = \left\{Q_1, \dots, Q_m\right\}\) if for each \(j \in \left\{1, \dots, l\right\}\), there exists \(i \in \left\{1, \dots, m\right\}\) such that \(P_j \subseteq Q_i\) (\cite{lee201716}). We write this \(P \succeq Q\) (\(\succ\) when the relation is strict). In anticipation of our later results, we note a tight connection between the refinement order and relative convexity of the value functions:
\begin{lemma}\label{finernecessitylemma}
    Let \(A\) and \(\hat{A}\) be finite. Then, \(\hat{V}-V\) is convex only if \(\hat{C} \succeq C\).
\end{lemma}
\begin{proof}
    Please visit Appendix \ref{finernecessitylemmaproof}.\end{proof}

It is easy to see that even if the agent only gains one additional action, when we compare the subdivisions \(C\) and \(\hat{C}\), we can say nothing in general about their relationship in the finer-than partial order. For instance, if the new action, \(\hat{a}\), strictly dominates all of the actions in \(A\), \(\hat{C}\) has a single cell, \(\Delta\left(\Theta\right)\), so \(C \succeq \hat{C}\) (Figure \ref{figsub3}). Moreover, the new action can be such that \(C\) and \(\hat{C}\) are incomparable (Figure \ref{figsub1}) or such that \(\hat{C} \succeq C\) (Figure \ref{figsub2}).

\begin{figure}
\centering
\begin{subfigure}{.5\textwidth}
  \centering
  \includegraphics[scale=.2]{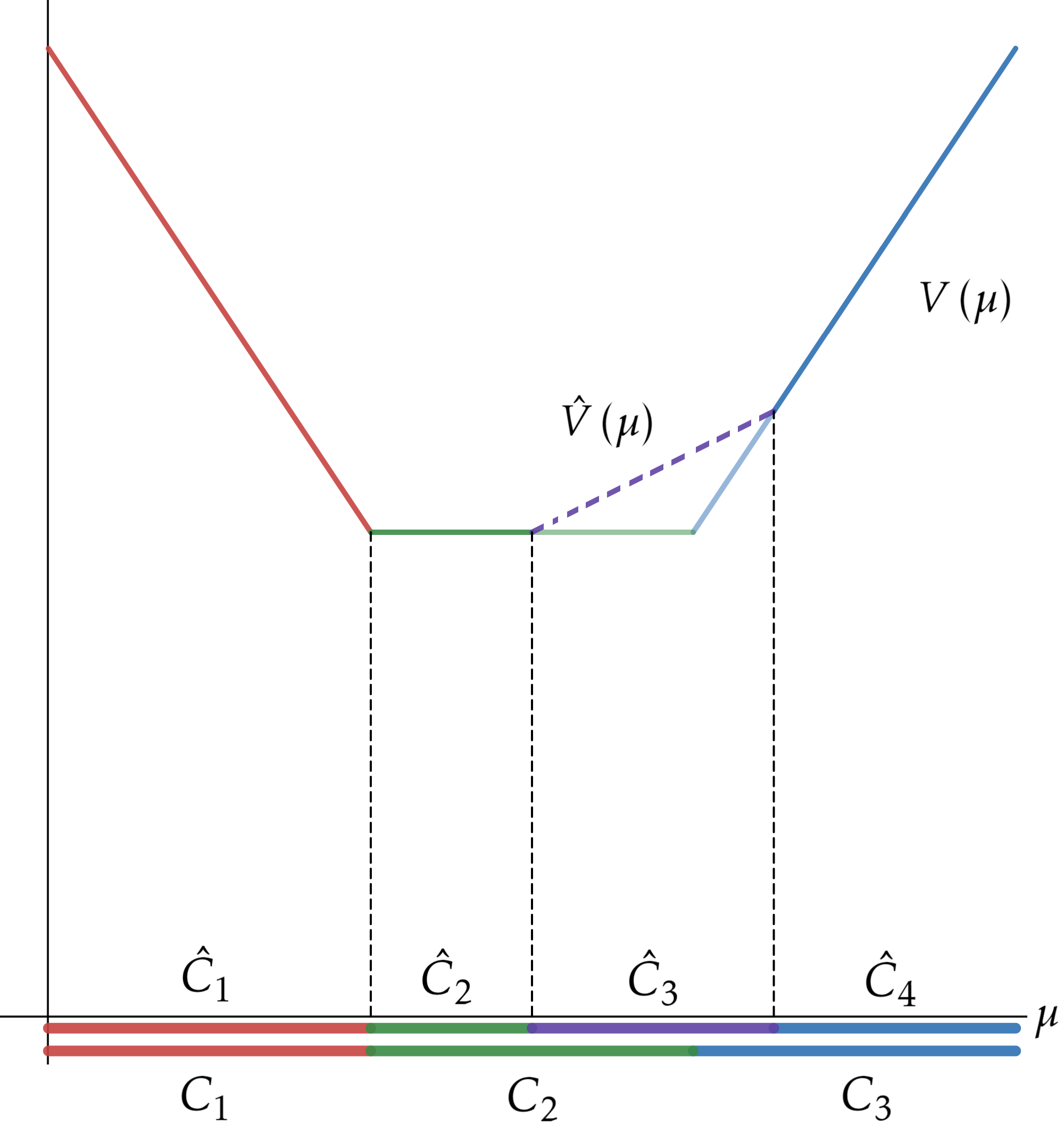}
  \caption{Incomparable \(C\) and \(\hat{C}\)}
  \label{figsub1}
\end{subfigure}%
\begin{subfigure}{.5\textwidth}
  \centering
  \includegraphics[scale=.2]{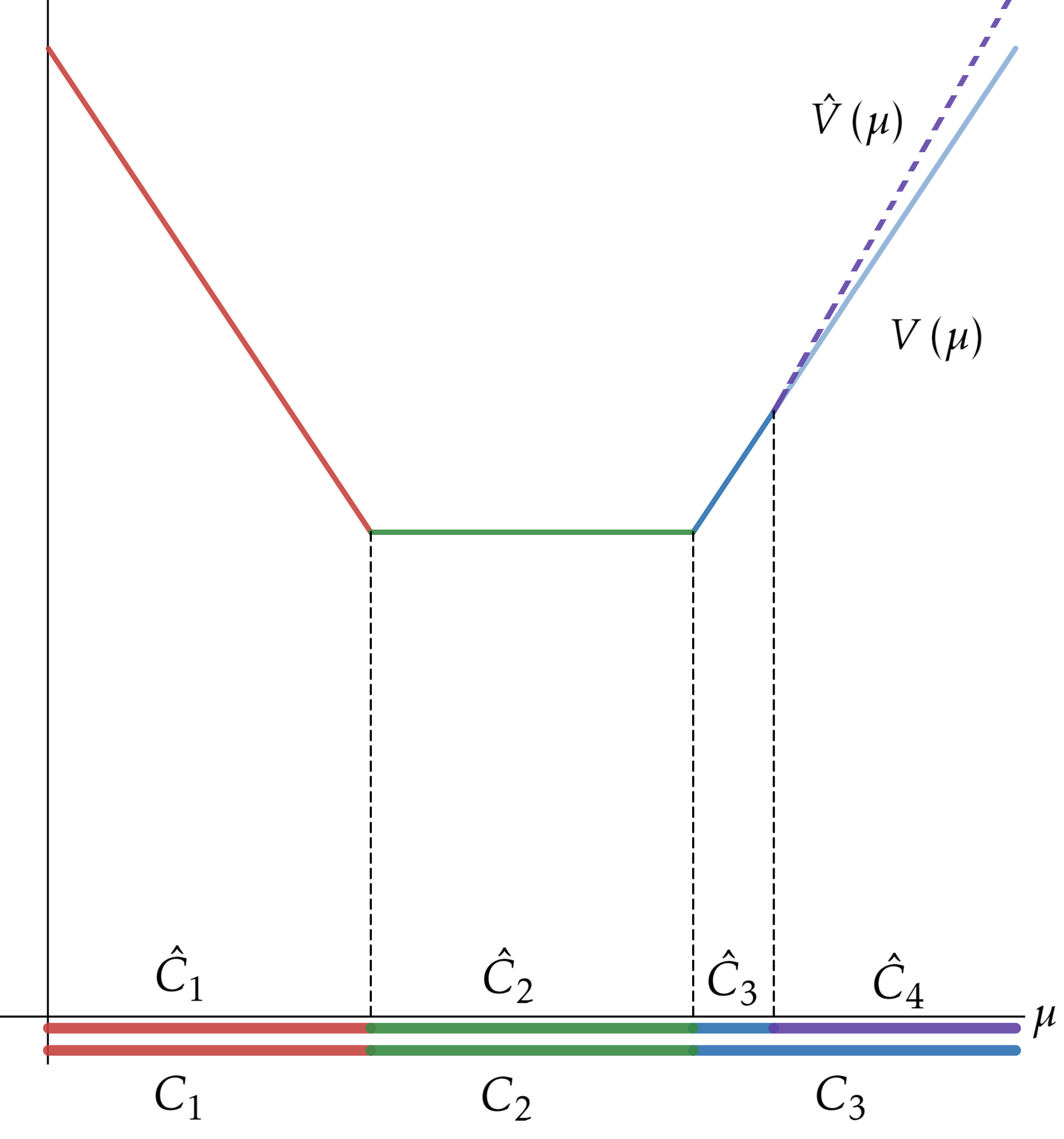}
  \caption{\(\hat{C} \succ C\)}
  \label{figsub2}
\end{subfigure}
\par
\bigskip
\par
\bigskip
\par
\begin{subfigure}{.5\textwidth}
  \centering
  \includegraphics[scale=.2]{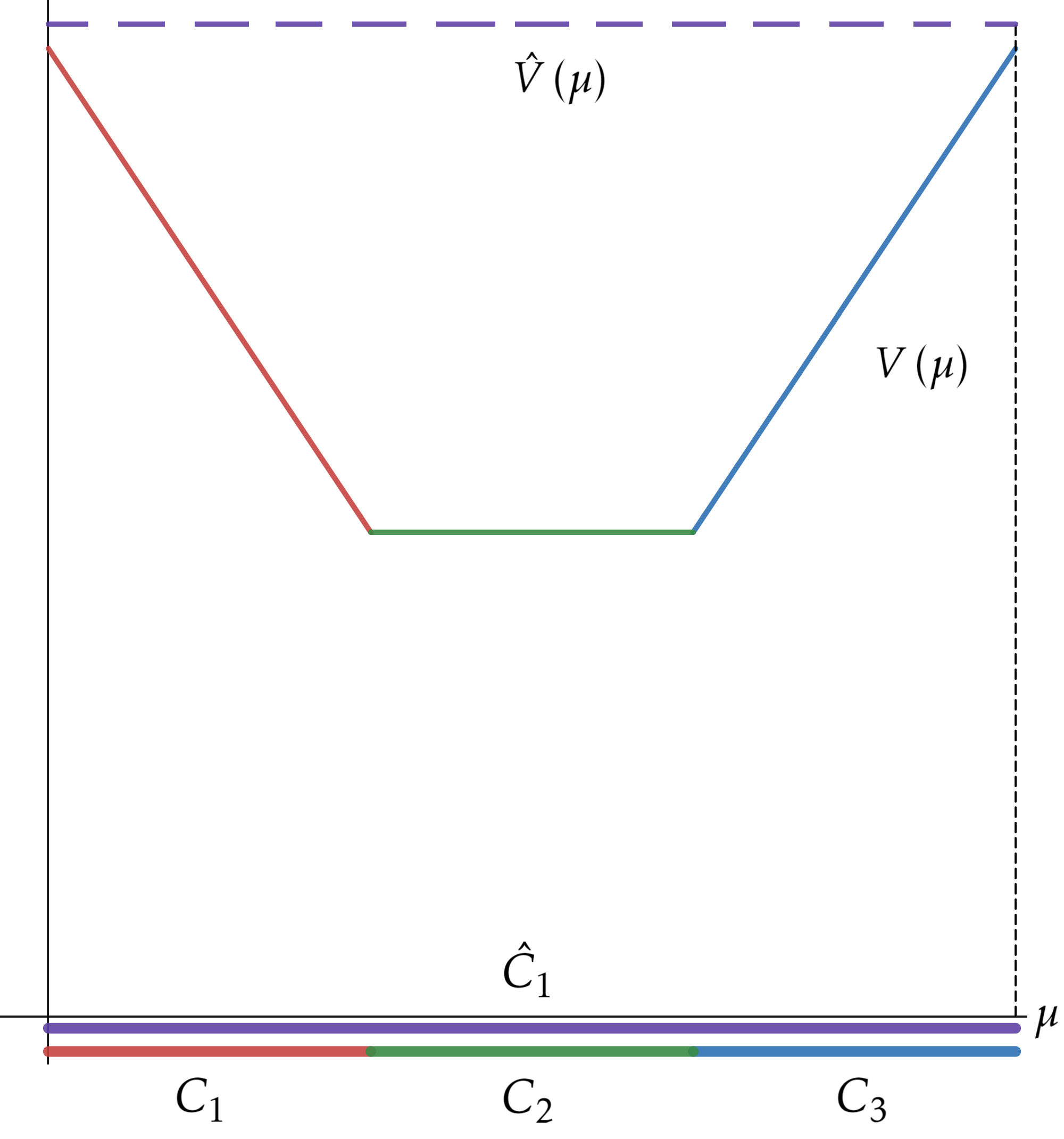}
  \caption{\(C \succ \hat{C}\)}
  \label{figsub3}
\end{subfigure} 
\caption{The three possible kinds of new action. There are two states and (initially) three actions. The payoff to the new, fourth, action is depicted in (dashed) purple. Only in \ref{figsub2} is the new action refining.}
\label{fig2}
\end{figure}

The last outcome is special. We say that a new action \(\hat{a}\) is \emph{Refining} if \(\hat{a}\) is not weakly dominated in \(A \cup \left\{\hat{a}\right\}\) and \(\hat{C} \succeq C\). That is, there exists \(\mu \in \Delta\left(\Theta\right)\) for which \(\mathbb{E}_{\mu}u\left(\hat{a},\theta\right) > V\left(\mu\right)\) and \[\left\{\mu \in \Delta\left(\Theta\right) \ \vert \ \mathbb{E}_{\mu}u\left(\hat{a},\theta\right) \geq V\left(\mu\right)\right\} \subseteq C_i\text{,}\]for some \(C_i \in C\). Refining actions are those that are good--uniquely optimal in at least one state of the world--but not too good--it is only one undominated action in \(A\) whose region of unique optimality shrinks as a result of adding the refining action. In other words, refining actions are precisely those whose addition to the decision problem lead to a positive change along the extensive margin--the agent must become more reactive.

To sum up, we began \(\S\)\ref{section2} by connecting an increase in an agent's value for information with an increase in the convexity of her value function. We then finished \(\S\)\ref{section2} with \(\S\)\ref{geometry}, where we observed that the means through which the transformation makes information more valuable can be decomposed into two channels. First, it must be that the agent becomes more sensitive or reactive to information--the extensive margin. Second, the value of distinguishing between different actions must also increase--the intensive margin. Moreover, we observed that the polyhedral subdivisions corresponding to the value functions capture any changes along the extensive margin exactly. In the next section, we make use of this representation to explore the effects of various alterations to the agent's decision problem on her value for information.

\section{Five Important Classes of Transformations}

In this section, we explore a variety of modifications to an agent's decision problem. Notably, we show that for several classes of transformation--adding and removing actions, and making the agent more or less risk-averse--increased sensitivity to information is also sufficient, or at least generically sufficient, for the transformation to make information more valuable. For these varieties, any transformation that does not make information more valuable must (generically) produce a negative change along the extensive margin.

\subsection{Becoming A Little More Flexible}\label{secmore}

We begin by transforming an agent's decision problem by adding a single action. We call this \emph{Making the Agent A Little More Flexible}. In this instantiation, we assume the set of actions initially available to the agent is finite and that in the transformed decision problem, the agent's utility function remains unchanged, \(\hat{u} = u\), but her new set of actions is \(\hat{A} \coloneqq A \cup \left\{\hat{a}\right\}\) for some \(\hat{a} \in \mathcal{A} \setminus A\). Our main result of this subsection reveals that subdivisions are central in understanding an agent's comparative value for information.
\begin{theorem}\label{moreflex}
Given finite \(\mathcal{D}\) and \(\hat{\mathcal{D}}\), obtained by making the agent a little more flexible, the following are equivalent:
    \begin{enumerate}[label={(\roman*)},noitemsep,topsep=0pt]
        \item \(\hat{a}\) is refining.
        \item The transformation does not generate less information acquisition.
        \item The transformation generates a greater value for information.
    \end{enumerate}
\end{theorem}

\medskip

This theorem is a consequence of Theorem \ref{moreconvex}, Lemma \ref{finernecessitylemma}, and Lemma \ref{convexandfiner}, below. Recall that Lemma \ref{finernecessitylemma} states that \(\hat{V}-V\) is convex only if \(\hat{C} \succeq C\), which can be understood as highlighting the necessity of increased sensitivity to information for information to become more valuable. But by definition, a refining action is precisely one that makes an agent more responsive to information--it, at most, ``replaces'' a single previous action. Accordingly, a refining action must produce an increase along the extensive margin. The other pertinent channel is the intensive margin. However, a refining action clearly must also improve things along the intensive margin: for any of the previously held actions, the value of distinguishing between them remains the same. Thus, adding a refining action must make information more valuable. The next lemma states this formally, and, therefore, produces the theorem.
\begin{lemma}\label{convexandfiner}
Given finite \(\mathcal{D}\) and \(\hat{\mathcal{D}}\), obtained by making the agent a little more flexible, \(\hat{V}-V\) is convex if \(\hat{C} \succeq C\).
\end{lemma}
\begin{proof}
    Please visit Appendix \ref{convexandfinerproof}.
\end{proof}


\subsection{Becoming More Flexible}\label{secmuch}
Now we transform an agent's decision problem by adding multiple actions. We call this \emph{Making the Agent More Flexible}. We maintain the assumptions that \(A\) is finite and that in the transformed decision problem, \(\hat{u} = u\). Now, her new set of actions is \(\hat{A} \coloneqq A \cup B\) for some additional finite set of actions \(B\in \mathcal{A} \setminus A\).

\begin{figure}
    \centering
    \includegraphics[scale=.25]{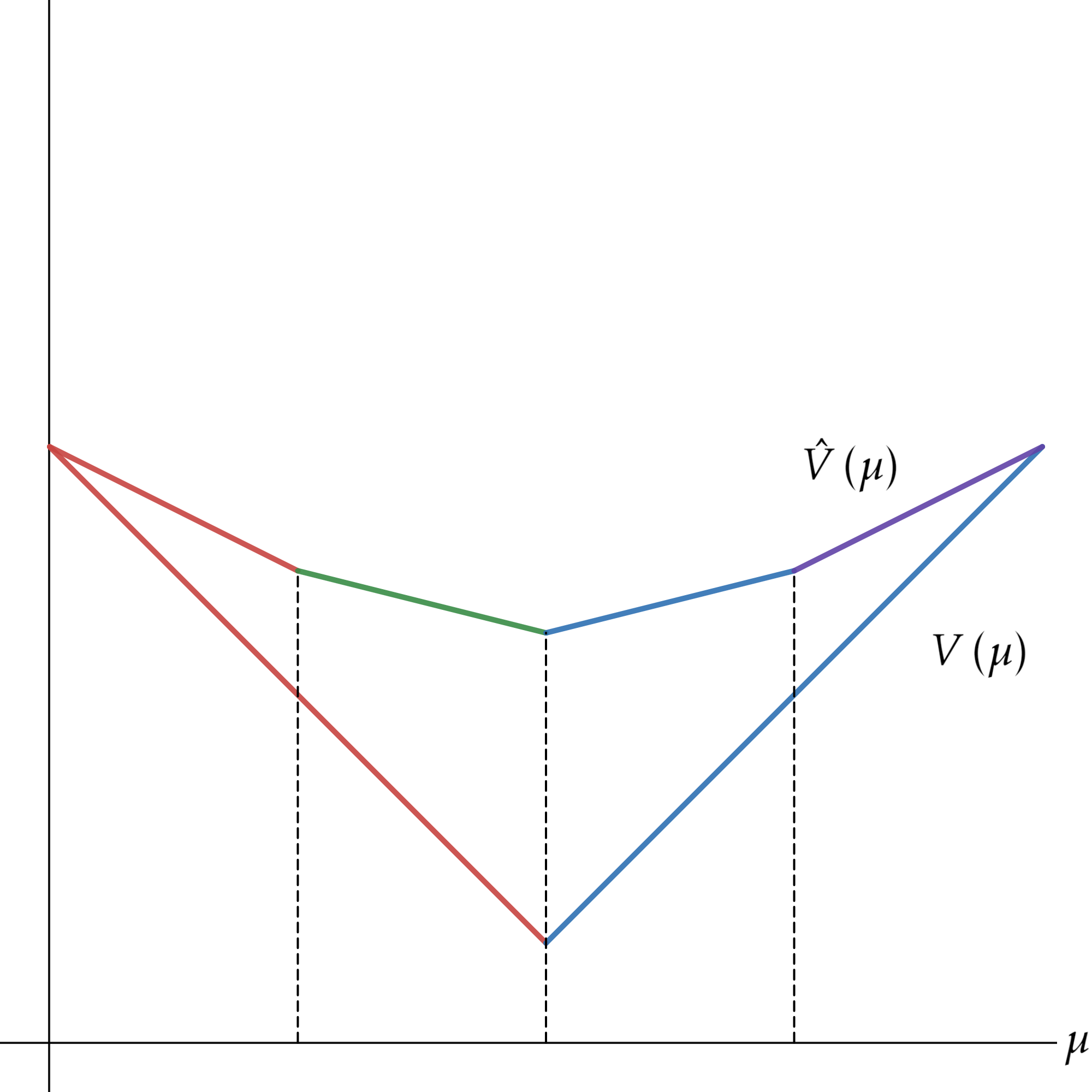}
    \caption{\(\hat{C} \succ C\) but \(\hat{V}-V\) is not convex.}
    \label{counterfigure1}
\end{figure}

With multiple new actions, we lose the equivalence of the convexity of \(\hat{V} - V\) and \(\hat{C}\)'s dominance of \(C\) in the refinement order. In particular, although Lemma \ref{finernecessitylemma} states that \(\hat{V} - V\) being convex implies \(\hat{C} \succeq C\), the converse is false when we make the agent more flexible. Intuitively, the value function \(\hat{V}\) can correspond to a finer subdivision than \(V\) but be shallower. The issue with adding multiple new actions is that now it is possible to ``replace'' previous sections of the value function in which the value of distinguishing between different actions is strictly lower. Such an occurrence is depicted in Figure \ref{counterfigure1}. So, even though the transformation makes information more valuable on the extensive margin, it produces a decrease on the intensive margin.

In order to preclude such an occurrence, we introduce another term. We say that a new set of actions is \emph{Totally Refining} if each \(b \in B\) is either weakly dominated or refining.

\begin{lemma}\label{seqextremal}
Given finite \(\mathcal{D}\) and \(\hat{\mathcal{D}}\), obtained by making the agent more flexible, \(\hat{V} - V\) is convex if \(B\) is totally refining.
\end{lemma}
\begin{proof}
For any \(b \in B\), let \(V_b\) denote the agent's value function when the set of actions is \(A \cup \left\{b\right\}\). Since each \(b\) is refining or weakly dominated, \(V_b - V\) is convex for all \(b \in B\). Finally, \(\hat{V} - V = \max\left(V_b\right)_{b \in B} - V = \max\left\{\left(V_b - V\right)_{b \in B}\right\}\) is convex. \end{proof}

Intuitively, any collection of refining actions must increase the agent's responsiveness to information. Theorem \ref{moreconvex} and Lemma \ref{seqextremal} produce
\begin{corollary}\label{seqeqcorr}
Given finite \(\mathcal{D}\) and \(\hat{\mathcal{D}}\), obtained by making the agent more flexible, the transformation generates a greater value for information and does not generate less information acquisition if the set of additional actions is totally refining.
\end{corollary}

The converse to Lemma \ref{seqextremal} is false, and Figure \ref{counterfig2} illustrates this. There, the agent gets access to two new actions, which increases her value of information. Note that the addition of these two actions leaves her subdivision unchanged, i.e., \(\hat{C} = C\) but \(\hat{V}-V\) is convex.\footnote{I am grateful to Gregorio Curello for suggesting this example.} Moreover, the addition of just one of these actions would not increase her value for information.
\begin{figure}
    \centering
    \includegraphics[scale=.25]{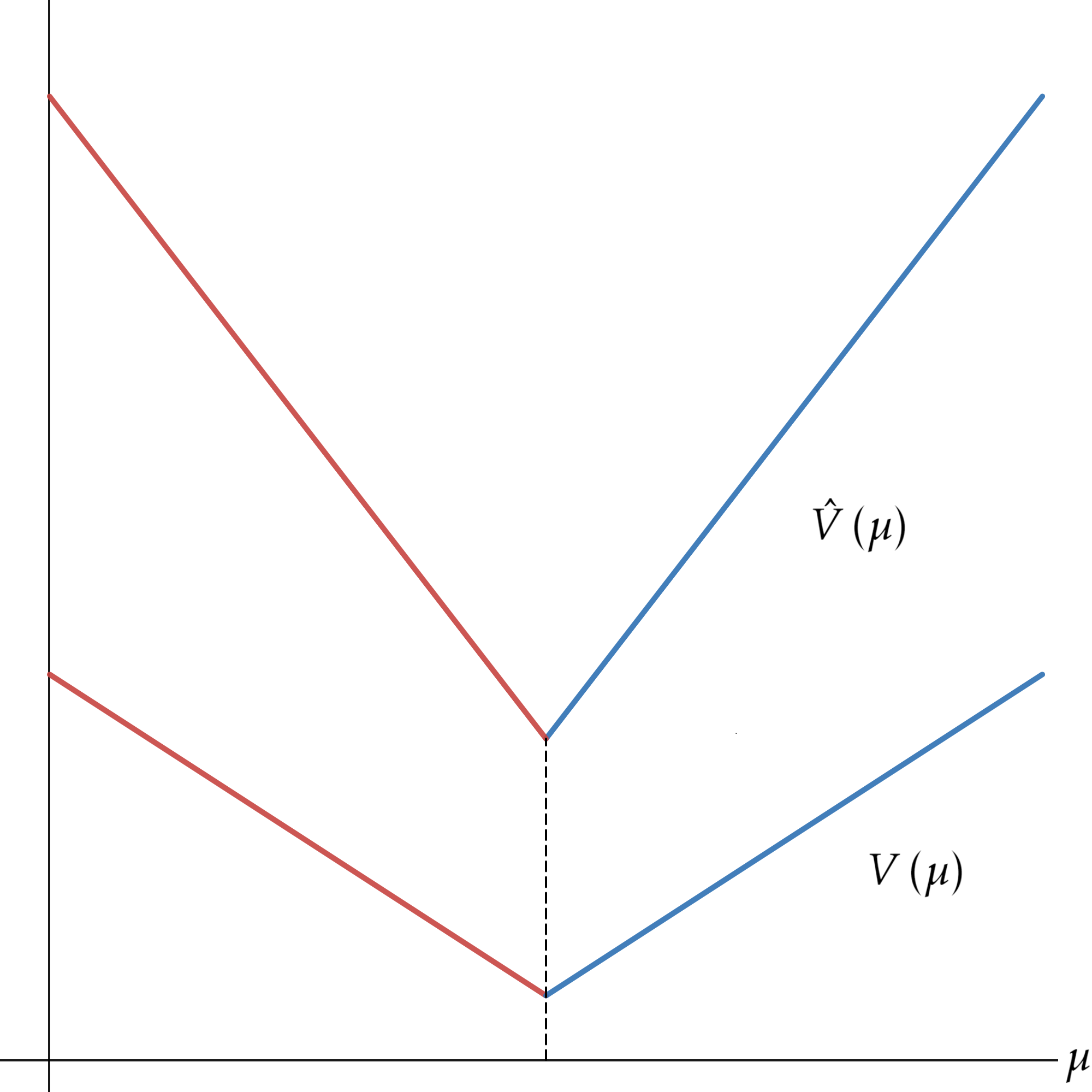}
    \caption{\(\hat{C} = C\) and \(\hat{V}-V\) is convex. However, \(B\) is not totally refining.}
    \label{counterfig2}
\end{figure}
On the other hand, the converse is almost true in the following sense.
\begin{figure}
\centering
\begin{subfigure}{.5\textwidth}
  \centering
  \includegraphics[scale=.18]{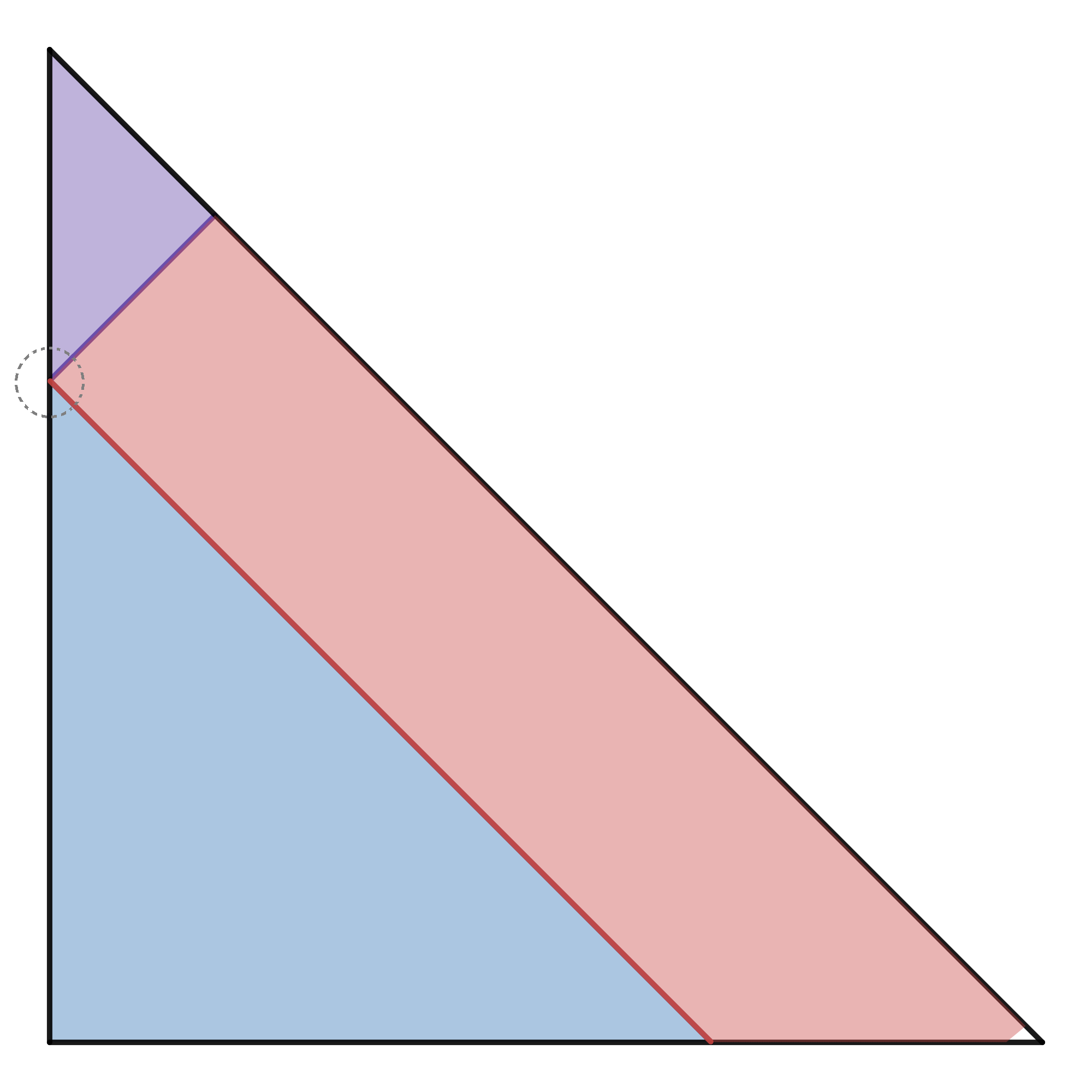}
  \caption{Not strictly refining.}
  \label{figsub0289}
\end{subfigure}%
\begin{subfigure}{.5\textwidth}
  \centering
  \includegraphics[scale=.18]{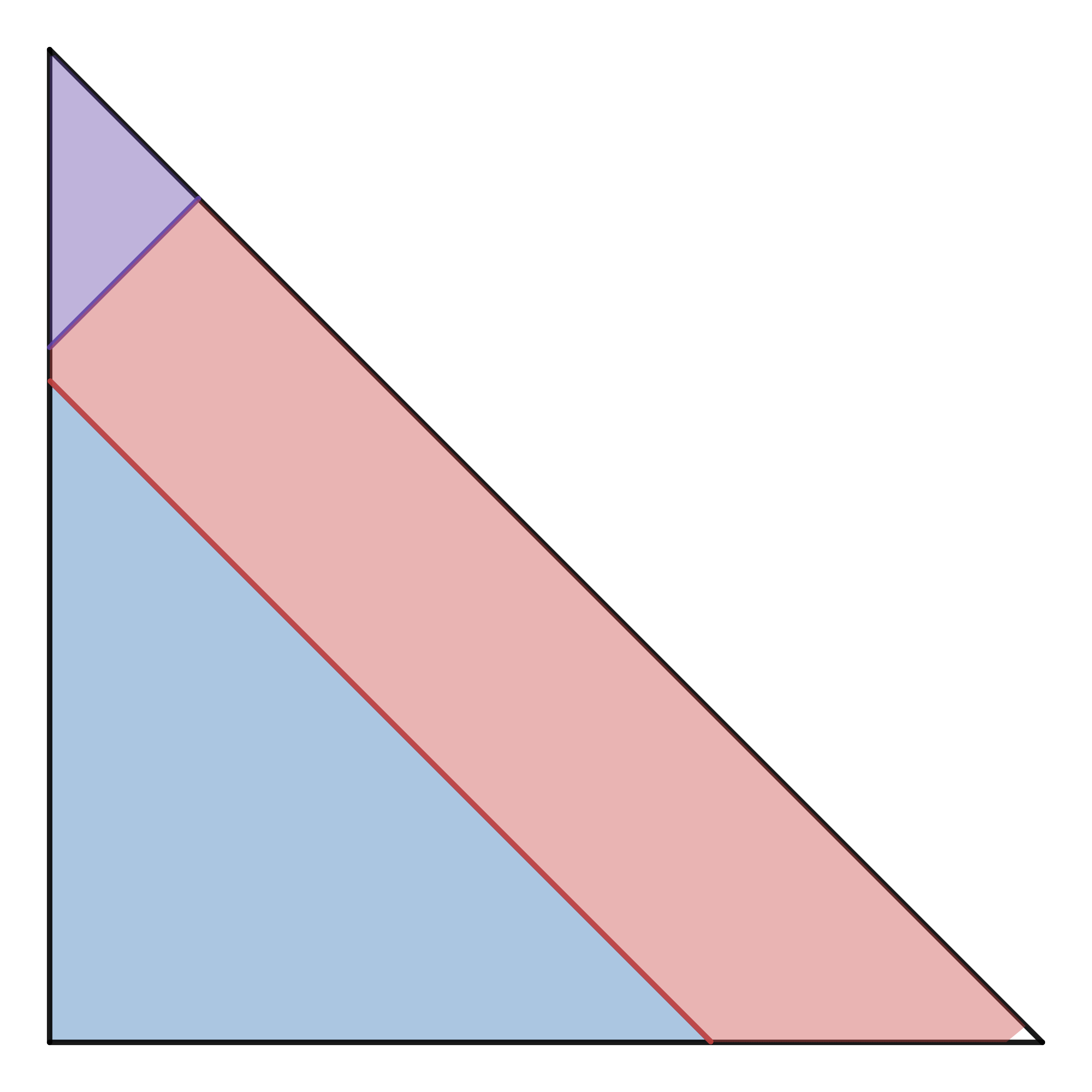}
  \caption{Strictly refining.}
  \label{figsub1289}
\end{subfigure}
\caption{The new action \(\hat{a}\) corresponds to the top left (purple) triangle.}
\label{fig1289}
\end{figure}

We say a new action \(\hat{a}\) is \emph{Strictly Refining} if the corresponding element \(\hat{C}_i \in \hat{C}\) is a subset of some element \(C_i \in C\)--i.e., is refining--and
\(\hat{C}_i \cap C_j = \emptyset\) for all \(C_j \in C \setminus \left\{C_i\right\}\). See Figure \ref{fig1289}. A new set of actions is \emph{Totally Strictly Refining} if each \(b \in B\) is either strictly dominated or strictly refining. For an additional set of actions \(B\), we understand the agent's utility to be an element, \(u\), of the Euclidean space \(\mathbb{R}^{B \times \Theta}\) (equipped with the Euclidean metric). Given an initial finite decision problem \(\mathcal{D}\) and a transformed finite decision problem \(\hat{\mathcal{D}}\), obtained by making the agent more flexible, we denote an agent's value function in the transformed decision problem by \(\hat{V}^u\) and say that \emph{The Transformation Generically Generates a Greater Value for Information and Does Not Generate Less Information} if \(\hat{V}^{\tilde{u}} - V\) is convex for all \(\tilde{u}\) in an open ball around \(u\).
\begin{proposition}\label{generic}
    Given finite \(\mathcal{D}\) and \(\hat{\mathcal{D}}\), obtained by making the agent more flexible, the transformation generically generates a greater value for information and does not generate less information acquisition if and only if the set of additional actions is totally strictly refining.
\end{proposition}
\begin{proof}
    The proof lies in Appendix \ref{genericproof}.
\end{proof}
This implies
\begin{corollary}
    Given finite \(\mathcal{D}\) and \(\hat{\mathcal{D}}\), obtained by making the agent more flexible, the transformation generically generates a greater value for information and does not generate less information acquisition only if the set of additional actions is totally refining.
\end{corollary}
The example depicted in Figure \ref{counterfig2} is to some extent a mirror of that depicted in Figure \ref{counterfigure1}. Like in the Figure \ref{counterfigure1} example, the agent in Figure \ref{counterfig2} is no less reactive to information after the transformation. However, each of the two new actions in the latter example, if added by itself, would alter the subdivision, thereby effecting a negative change on the extensive margin and so would not make information more valuable. Proposition \ref{generic} reassures us that this example is special. If a collection of actions, not all of which are refining, makes the agent more sensitive to information, the actions must be fine-tuned to the decision problem at hand.

\subsection{Becoming Less Flexible}\label{lessflexsec}

The generosity of the previous subsections has come to an end. Now, we transform the agent's decision problem by \textit{removing} actions. We call this \emph{Making the Agent Less Flexible}. As in \(\S\)\ref{secmore} and \(\S\)\ref{secmuch}, \(A\) is finite and \(\hat{u} = u\). In contrast, in the transformed decision problem, the agent's new set of actions is \(\emptyset \neq \hat{A} \subset A\).

If there exists an action \(a_i \in \hat{A}\) that is not weakly dominated when the set of actions is \(A\)--there exists an \(\mu \in \Delta \left(\Theta\right)\) such that \(\mathbb{E}_{\mu}u\left(a_i,\theta\right) > \max_{a \in A \setminus \left\{a_i\right\}}\mathbb{E}_{\mu}u\left(a,\theta\right)\)--we say \textit{There Are Leftovers}. If \(\hat{V} \neq V\), we say the elimination is \textit{Consequential}. An elimination that is not consequential is \textit{Inconsequential}.
\begin{proposition}\label{lessflexprop}
    Given finite \(\mathcal{D}\) and \(\hat{\mathcal{D}}\), obtained by making the agent less flexible, if the transformation generates a greater value for information and generates no less information acquisition, there are no leftovers or the removal is inconsequential.
\end{proposition}
\begin{proof}
    The proof lies in Appendix \ref{lessflexpropproof}.
\end{proof}
Clearly, if the pruning of actions is inconsequential, making the agent less flexible cannot reduce her value for information: her value function is unchanged by the removal of dominated actions. Eliminating actions that are not dominated may yet lead to an increase in her value of information, but only if, in effect, the agent obtains an entirely new decision problem. An increase in her value for information can only occur (when some undominated actions are removed), if every undominated action is removed--there can be no leftovers.

Even worse is the situation in which \(A \setminus \hat{A}\) is totally refining with respect to \(\hat{A}\)--\textit{viz.,} starting with set \(\hat{A}\), each \(a \in A \setminus \hat{A}\) is refining (they are not weakly dominated, by assumption). Then, as Lemma \ref{seqextremal} implies \(V-\hat{V}\) is convex, 
\begin{corollary}Given finite \(\mathcal{D}\) and \(\hat{\mathcal{D}}\), obtained by making the agent less flexible, if \(A \setminus \hat{A}\) is totally refining with respect to \(\hat{A}\), the transformation generates a lower value for information and generates no more information acquisition.
\end{corollary}
That is 
\[\mathbb{E}_{\Phi}\hat{V}\left(\mu\right) - \hat{V}\left(\mu_{0}\right) \leq \mathbb{E}_{\Phi}V\left(\mu\right) - V\left(\mu_{0}\right)\text{,}\] for all \(\Phi \in \mathcal{F}\left(\mu_{0}\right)\) and \(\mu_{0} \in \inter \Delta\left(\Theta\right)\);
and for any solution to Problem \ref{flexhat}, \(\hat{\Phi}^*\), there exists a solution to Problem \ref{flex}, \(\Phi^*\), that is not a strict MPC of \(\hat{\Phi}^*\).

It is easiest to understand Proposition \ref{lessflexprop} by reflecting upon the removal of a single action. For information to be more valuable, we know that the new subdivision must be finer than the original one--the agent must be more responsive to information. But if the action being removed was strictly optimal at some belief, its removal must make the agent strictly \textit{less} responsive to information. This can be seen most starkly when the initial set of actions contained just two undominated actions. If one is removed, any information is now worthless.

\subsection{Affine Transformations of the Agent's Utility Function}\label{transsec}

Now, the transformation leaves the set of actions available to the agent unchanged; that is, \(\hat{A} = A\). Instead, in the transformed decision problem, the agent's new utility is \(\hat{u} \coloneqq k u + s\), where \(k \in \mathbb{R}_{++}\) and \(s \in \mathbb{R}\). Such a transformation preserves the subdivision--\(\hat{C} = C\)--but how does it affect \(\hat{V} - V\)?
\begin{proposition}\label{affinetrans} 
Given \(\mathcal{D}\), an affine transformation of the agent's utility generates a greater value for information and does not generate less information acquisition if and only if \(k \geq 1\).\footnote{\cite{denti2022posterior} also notes this result, but does not give a proof.}
\end{proposition}
\begin{proof}
    The proof resides in Appendix \ref{affinetransproof}.
\end{proof}
When the agent's utility is affinely-transformed, the extensive-margin channel through which the transformation makes information more valuable is shut down. The agent stays precisely as sensitive to information as she had been before. Instead, an affine transformation makes information more valuable (or does not) via the intensive margin. Proposition \ref{affinetrans}, then, is obvious: the value of distinguishing between actions increases if and only if the slopes (in belief space) of the actions' payoffs become steeper.

Affine transformations can be effected in a number of ways. The first is a direct scaling of the payoffs. The second is by repeating the decision problem \(k \in \mathbb{N}\) times.
\begin{corollary}
    Repetition of a decision problem (with future payoffs possibly discounted) generates a greater value for information and does not generate less information acquisition.
\end{corollary}
A third relates to changes in wealth. The agent's monetary payoff from action \(a\) is some function \(f_a\left(\theta\right)\). Her utility function over terminal wealth \(w\) is of the constant absolute risk aversion (CARA) form: \(v\left(w\right) = - \exp\left(-\alpha w\right)\), where \(\alpha \in \mathbb{R}_{++}\) is a scaling parameter. Accordingly, the agent's utility as a function of her action and the state \(\theta\) is
\[u\left(a,\theta\right) = - \exp\left(-\alpha f_a\left(\theta\right)\right) \exp\left(-\alpha w\right)\text{.}\]  
Evidently, changing the agent's endowed wealth to \(\hat{w}\) produces linear transformation of the utility \(u \mapsto k u = \hat{u}\), where
\[k = \frac{\exp\left(-\alpha \hat{w}\right)}{\exp\left(-\alpha w\right)}\text{.}\]
Consequently, 
\begin{corollary}
    For an agent with CARA utility, decreased wealth generates a greater value for information and does not generate less information acquisition. 
\end{corollary}
We could also assume that the agent's endowed wealth is random, modeled by a real valued, random variable \(Y\) that has a finite mean and is uncorrelated with the state. \(Y\) is distributed according to cumulative distribution function \(H\). We say that \emph{Aggregate Risk Increases} if the distribution of \(Y\) changes from \(H\) to an MPS of \(H\), \(\hat{H}\). As this increase in risk is equivalent to a decrease in wealth for the risk-averse agent, 
\begin{corollary}
    For an agent with CARA utility, increased aggregate risk generates a greater value for information and does not generate less information acquisition. 
\end{corollary}

\subsection{Increases (or Decreases) in Risk Aversion}\label{riskaverse}

We now turn our attention to two particular \textit{non-linear} transformations of the agent's utility function. We specialize to the scenario in which the set of actions available to the agent, \(A\), is finite. We say the agent \emph{Becomes More Risk Averse} if in the transformed decision problem \(\hat{u} = \phi \circ u\) for some strictly increasing, strictly concave \(\phi \colon \mathbb{R} \to \mathbb{R}\). Similarly, the agent \emph{Becomes More Risk Loving} if \(\hat{u} = \phi \circ u\) for some strictly increasing, strictly convex \(\phi \colon \mathbb{R} \to \mathbb{R}\).
\begin{proposition}\label{riskaverseprop}
    Given finite \(\mathcal{D}\), making the agent more risk averse or more risk loving does not generically generate a greater value for information and generically may generate less information acquisition.
\end{proposition}
We prove this result by showing that if making the agent more risk averse (or risk loving) leads to a finer subdivision, we can always find an arbitrarily small perturbation of the agent's utility that destroys this refinement. This result is to be expected: as revealed in \cite{weinstein2016effect} and \cite*{battigalli2016note}, previously dominated actions may become undominated as an agent becomes more risk averse. In another related paper, \cite{safety} characterize precisely when the set of beliefs at which one action is preferred to another must increase in size (in a set-inclusion sense)--thereby altering the subdivision--when an agent is made more risk averse.

As with the earlier transformations (except for affine ones), we see that whether increased (or decreased) risk aversion makes information more valuable generically depends entirely on the extensive margin change produced by the transformation. Furthermore, our answer is negative: for generic decision problems a change in the agent's risk aversion cannot lead to an improvement along the extensive margin. The agent cannot become more sensitive to information except in knife-edge instances.

\section{Fixed-Mean Comparisons}

Our two earlier conceptions of making information more valuable were demanding: we required information to become more valuable no matter what the prior was and regardless of what the information or the information-acquisition technology was. Now we require that information become more valuable for a fixed prior.

\begin{definition}[Exogenous Information]
    Fix a prior \(\mu_0 \in \inter \Delta\left(\Theta\right)\). Given \(\mathcal{D}\) and \(\hat{\mathcal{D}}\), we say that \emph{The Transformation \(\mu_0\)-Generates a Greater Value for Information} if \[\mathbb{E}_{\Phi}\hat{V}\left(\mu\right) - \hat{V}\left(\textcolor{OrangeRed}{\mu_0}\right) \geq \mathbb{E}_{\Phi}V\left(\mu\right) - V\left(\textcolor{OrangeRed}{\mu_0}\right)\text{,}\] for all \(\Phi \in \mathcal{F}\left(\textcolor{OrangeRed}{\mu_0}\right)\).
\end{definition}

This definition still requires that the value of information increases for the agent no matter the information (for any Bayes-plausible \(\Phi\)), but now the prior is fixed.

\begin{definition}[Endogenous Information]
    Fix a prior \(\mu_0 \in \inter \Delta\left(\Theta\right)\). Given \(\mathcal{D}\) and \(\hat{\mathcal{D}}\), we say that \emph{The Transformation Does Not \(\mu_0\)-Generate Less Information Acquisition} if for any UPS cost functional \(D\), and solution to the agent's information acquisition problem in the initial decision problem (Problem \ref{flex}), \(\Phi^{*}\), there exists an solution to the agent's information acquisition problem in the transformed decision problem (Problem \ref{flexhat}), \(\hat{\Phi}^{*}\), that is not a strict mean-preserving contraction (MPC) of \(\Phi^{*}\).
\end{definition}
\begin{figure}
    \centering
    \includegraphics[width=\textwidth]{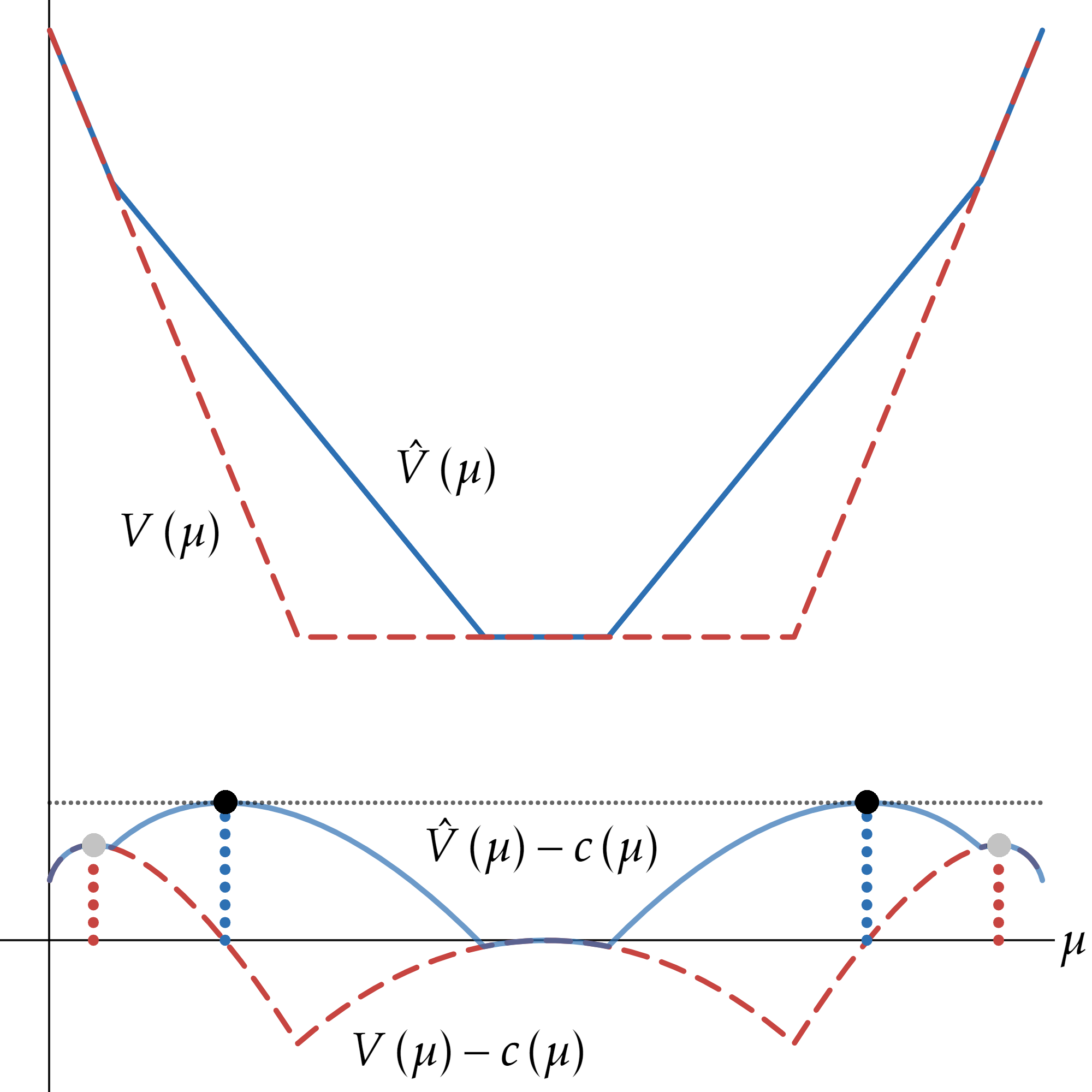}
    \caption{When her prior equals \(1/2\), the agent strictly acquires less information as a result of the transformation.}
    \label{figcounterfix}
\end{figure}

Our first observation is that the equivalence between the endogenous- and exogenous-information formalizations of making information more valuable stated for arbitrary priors in Theorem \ref{moreconvex} does not persist in the fixed-prior environment of this section. Figure \ref{figcounterfix} illustrates \(V\) and \(V - c\) (dashed red) and \(\hat{V}\) and \(\hat{V}-c\) (solid blue), where \(c\) is a strictly convex function. \(\hat{V}\) is obtained by adding two non-refining actions to the action set in \(\mathcal{D}\). 

The agent's prior is \(1/2\). The black points indicate the support of the agent's uniquely optimal information acquisition in \(\hat{\mathcal{D}}\), \(\hat{\Phi}\). The light grey points indicate the support of the agent's uniquely optimal information acquisition in \(\mathcal{D}\), \(\Phi\). As both supports are affinely independent and \(\supp \hat{\Phi} \subset \inter \conv \supp \Phi\) (the support of \(\hat{\Phi}\) is a strict subset of the relative interior of the convex hull of the support of \(\Phi\)), the agent acquires strictly more information pre-transformation. Yet, the transformation \(\mu_0\)-generates a greater value for information: \(\hat{V}\left(\mu_0\right) = V(\mu_0)\) and \(\hat{V}\) lies above \(V\) pointwise, making any Bayes-plausible \(\Phi\) more valuable in the transformed decision problem.

This example also highlights that the fixed- and free-mean notions of information becoming more valuable are not equivalent, even when information is exogenous. In the example \(\hat{V}-V\) is not convex, so the transformation does not generate a greater value for information, yet it \(\mu_0\)-generates a greater value for information when \(\mu_0 = 1/2\). The example suggests that pointwise dominance of \(\hat{V}\) over \(V\) is important, and as we will now show, a generalized version of this is key.

We say that value function \(\hat{V}\) \emph{Shift-Majorizes} value function \(V\) if there exists an affine function \(\ell\left(\mu\right) \coloneqq \lambda \cdot \mu + \tau\) (where \(\lambda \in \mathbb{R}^{n-1}\) and \(\tau \in \mathbb{R}\)) satisfying
\begin{enumerate}[label={(\roman*)},noitemsep,topsep=0pt]
    \item Equality at the prior: \(\hat{V}\left(\mu_{0}\right) + \ell\left(\mu_{0}\right) = V\left(\mu_{0}\right)\); and
    \item Pointwise dominance on \(\Delta\left(\Theta\right)\): \(\hat{V}\left(\mu\right) + \ell\left(\mu\right) \geq V\left(\mu\right)\) for all \(\mu \in \Delta\left(\Theta\right)\).
\end{enumerate}

\begin{theorem}\label{mugenerate}
    Fix a prior \(\mu_0 \in \inter \Delta\left(\Theta\right)\). Given \(\mathcal{D}\) and \(\hat{\mathcal{D}}\), the transformation \(\mu_0\)-generates a greater value for information if and only if \(\hat{V}\) shift-majorizes \(V\).
\end{theorem}
\begin{proof}
    The proof lies in Appendix \ref{mugenerateproof}.
\end{proof}
Shift-majorization is easy to check. In particular, if both \(V\) and \(\hat{V}\) are differentiable at \(\mu_0\)--which is a generic property of convex functions (Theorem 25.5 in \cite{rockafellar2015convex})--\(\ell\) is simply the difference of the tangent hyperplanes to \(\hat{V}\) and \(V\) at \(\mu_0\) (which are unique). Moreover, when information is endogenous, shift-majorization of \(\hat{V}\) over \(V\) is necessary, though, as Figure \ref{figcounterfix} reveals, it is not sufficient.
\begin{theorem}\label{muless}
Fix a prior \(\mu_0 \in \inter \Delta\left(\Theta\right)\). Given \(\mathcal{D}\) and \(\hat{\mathcal{D}}\), the transformation does not \(\mu_0\)-generate less information acquisition only if \(\hat{V}\) shift-majorizes value function \(V\).
\end{theorem}
\begin{proof}
    Please see Appendix \ref{mulessproof}.
\end{proof}

Combining these theorems produces
\begin{corollary}\label{corrgen}
    Fix a prior  \(\mu_0 \in \inter \Delta\left(\Theta\right)\). Given \(\mathcal{D}\) and \(\hat{\mathcal{D}}\), if the transformation does not \(\mu_{0}\)-generate less information acquisition, it \(\mu_0\)-generates a greater value for information.
\end{corollary}
That is, increasing the value of endogenous information (for a fixed prior) always increases the value of exogenous information for a fixed prior. Before going on to study adding and subtracting actions in the fixed-mean paradigm, let us add one further result, which connects the fixed- and free-mean results.
\begin{definition}
    Fix a prior \(\mu_0 \in \inter \Delta\left(\Theta\right)\). Given \(\mathcal{D}\) and \(\hat{\mathcal{D}}\), we say that \emph{The Transformation \(\mu_0\)-Generates a Greater Incremental Value for Information} if \[\mathbb{E}_{\Phi}\hat{V}\left(\mu\right) - \mathbb{E}_{\Upsilon}\hat{V}\left(\mu\right) \geq \mathbb{E}_{\Phi}V\left(\mu\right) - \mathbb{E}_{\Upsilon}\textcolor{OrangeRed}{V}\left(\mu\right)\text{,}\] for all \(\Phi, \Upsilon \in \mathcal{F}\left(\textcolor{OrangeRed}{\mu_0}\right)\) with \(\Upsilon \in MPC(\Phi)\).
\end{definition}
Then,
\begin{corollary}\label{incremental}
    Fix a prior \(\mu_0 \in \inter \Delta\left(\Theta\right)\). Given \(\mathcal{D}\) and \(\hat{\mathcal{D}}\), the transformation \(\mu_0\)-generates a greater incremental value for information if and only if \(\hat{V}-V\) is convex.
\end{corollary}
\begin{proof}
    The result is an easy consequence of Lemmas \ref{lemmany} and \ref{newexoginfo}. The detailed proof lies in Appendix \ref{incrementalproof}.
\end{proof}

\subsection{Becoming More or Less Flexible, Redux}

We now specialize to the situations in which we are either adding or subtracting actions, \textit{viz.}, making the agent more or less flexible, respectively. We stipulate that the initial set of actions, \(A\), available to the agent is finite and either add finitely more actions \(B \subseteq \mathcal{A} \setminus A\) or remove finitely many actions.

Let \(A\left(\mu_{0}\right)\) denote the set of actions that are optimal at the prior \(\mu_{0}\) in \(\mathcal{D}\). Analogously, let \(\hat{A}\left(\mu_{0}\right)\) denote the set of actions that are optimal at the prior \(\mu_{0}\) in \(\hat{\mathcal{D}}\). We say that \emph{Some Action Remains Prior-Optimal} if \(A\left(\mu_{0}\right) \cap \hat{A}\left(\mu_{0}\right) \neq \emptyset\). Thus, \emph{No Action Remains Prior-Optimal} if \(A\left(\mu_{0}\right) \cap \hat{A}\left(\mu_{0}\right) = \emptyset\). If the optimal action at \(\mu_0\) in \(\mathcal{D}\) is unique, we denote it \(a_{\mu_{0}}\). Likewise, if the optimal action at \(\mu_0\) in \(\hat{\mathcal{D}}\) is unique, we denote it \(\hat{a}_{\mu_{0}}\).
\begin{proposition}\label{prioroptimal}
    Fix a prior \(\mu_0 \in \inter \Delta\left(\Theta\right)\) and take finite \(\mathcal{D}\) and \(\hat{\mathcal{D}}\).
    \begin{enumerate}[label={(\roman*)},noitemsep,topsep=0pt]
        \item\label{statement1} Suppose \(\hat{\mathcal{D}}\) is obtained by making the agent more flexible. The transformation \(\mu_0\)-generates a greater value for information if some action remains prior-optimal.
        \item Suppose \(\hat{\mathcal{D}}\) is obtained by making the agent less flexible and the removal is consequential.
        \begin{enumerate}[noitemsep,topsep=0pt]
        \item\label{statement2} The transformation \(\mu_0\)-generates a greater value for information only if no action remains prior-optimal.
        \item\label{statement3} The transformation does not \(\mu_0\)-generate less information acquisition only if no action remains prior-optimal.
    \end{enumerate}
    \end{enumerate}
\end{proposition}
\begin{proof}
    Appendix \ref{prioroptimalproof} contains the proof.\end{proof}
If no action remains prior-optimal, the effect of making the agent more flexible is ambiguous, though, as we will now see, the deleterious effect of deleting actions observed in \(\S\)\ref{lessflexsec} persists. We say that prior \(\mu_0\) is \emph{Generic} if \(\mu_0 \in \textcolor{OrangeRed}{\inter \left(C_i\right) \cap \inter \left(\hat{C}_j\right)}\) for some \(C_i \in C\) and \(\hat{C}_j \in \hat{C}\). If \(\mu_0\) is generic, the optimal actions in \(\mathcal{D}\) and \(\hat{\mathcal{D}}\) at \(\mu_0\) are unique. Our next result states that is only if the prior-optimal action \(\hat{a}_{\mu_{0}}\) in \(\hat{\mathcal{D}}\) is weakly dominated by the previously prior-optimal action, \(a_{\mu_0}\), in \(\mathcal{D}\) that removing actions can increase the value of information for generic priors. 
\begin{proposition}\label{badremoval}
    Fix a generic prior \(\mu_0 \in \inter \Delta\left(\Theta\right)\), and take finite \(\mathcal{D}\) and \(\hat{\mathcal{D}}\), obtained by making the agent less flexible. If \(\hat{V}\) shift-majorizes \(V\), the elimination is inconsequential (so \(\hat{V} = V\)) or \(\hat{a}_{\mu_{0}}\) is weakly dominated by \(a_{\mu_0}\) in \(\mathcal{D}\).
\end{proposition}
\begin{proof}
    The proof lies in Appendix \ref{badremovalproof}.
\end{proof}
This pair of propositions is reminiscent of Proposition \ref{lessflexprop}, which studies the removal of actions when the prior is not fixed. There, we find that any consequential removal of actions--that does not result in effectively a brand new decision problem--cannot lead to a more convex value function and so cannot increase the value of information. Moreover, when actions are removed, things are ``doubly'' bad as the agent's payoff is necessarily lower (she is optimizing over a smaller set). This direct consequence is more important with a fixed prior and we see that it leads to a similarly stark outcome: for a generic prior, any consequential removal of actions must lead to a lower value for information, except in the special case when there is a new action that is now optimal at the prior in \(\hat{\mathcal{D}}\) that had been previously dominated by the prior-optimal action in \(\mathcal{D}\). In short, except in special cases, irrespective of whether the prior is fixed, reducing the agent's ability to react to the world in a meaningful way lowers her value for information.

\section{Two Applications}\label{section4}

\subsection{Delegation}

In a delegation setting, in which an agent acquires information before taking an action, \cite*{clearadvice} studies how a principal prefers to constrain the agent's set of actions even though their interests are perfectly aligned \textit{ex post}. In particular, \citeauthor{clearadvice} shows that it is optimal for the principal to eliminate ``intermediate'' actions, which improves incentives for information acquisition. In that spirit, here we note that when the agent chooses whether to buy a fixed experiment, the principal always finds it optimal to increase the agent's flexibility by giving him additional refining actions.

Suppose the principal and agent share the same utility function, a common prior, and that initially the set of actions available to the agent is the finite set \(A\). The agent can acquire information by paying some cost \(\gamma > 0\) to see the realization of some signal. After acquiring information, the agent takes an action. The principal can give the agent access to an additional finite set of actions, \(B\), before she acquires information.
\begin{remark}
The principal prefers to give the agent access to an additional set of actions, \(B\), if it is totally refining.
\end{remark}
However, when \(A\) is such that no action is weakly dominated, Propositions \ref{lessflexprop}, \ref{prioroptimal}, and \ref{badremoval} reveal that regardless of whether the agent has private information (the principal knows her prior), the principal cannot increase the value of information for the agent by removing actions. Moreover, the direct effect of removing actions is also negative. Thus, any removal of actions, even intermediate ones, may result in a strictly lower payoff for the principal.

This observation seemingly contradicts the thesis of \cite{clearadvice}. Not so: the crucial features of \cite{clearadvice} are that the principal knows the agent's information-acquisition technology (which takes a specific fixed form), he knows her prior, and the set of actions is rich (an interval).\footnote{As the action set is not finite, Proposition \ref{badremoval} does not apply.} A key aspect of our formalization of an increased value for information is that any kind of exogenous information must be more valuable for the agent or that the agent cannot prefer to learn less no matter her posterior-separable cost. In \cite{clearadvice}, the removal of actions can be fine-tuned to the cost and the prior; whereas, here, the cost can, in a sense, be fine-tuned to the removal.

\subsection{Selling Information}\label{sell}

Theorem \ref{moreconvex} suggests a natural analog of increasing differences in an informational setting in which an agent's private type is her value for information. Here we apply this to a monopolistic screening problem. There are \(n\) states and an agent has one of two types, \(\omega_1\) and \(\omega_2\), with respective value functions \(V_{1}\) and \(V_{2}\), where \(V_{1} - V_{2}\) is convex. The principal and agent share a common prior \(\mu_{0} \in \inter{\Delta\left(\Theta\right)}\) and the principal can ``produce'' any distribution over posteriors \(\Phi\) subject to a UPS cost \(D\left(\Phi\right)\). By the revelation principle, she offers a contract \(\left(\left(t_1, \Phi_1\right), \left(t_2, \Phi_2\right)\right)\).

Naturally, in the first-best problem, the principal solves
\[\max_{\Phi_1 \in \mathcal{F}\left(\mu_{0}\right)}\left\{\int_{\Delta\left(\Theta\right)}V_1\left(\mu\right)d\Phi_1\left(\mu\right) - D\left(\Phi_{1}\right)\right\} \text{,} \quad \text{and} \quad \max_{\Phi_2 \in \mathcal{F}\left(\mu_{0}\right)}\left\{\int_{\Delta\left(\Theta\right)}V_2\left(\mu\right)d\Phi_2\left(\mu\right) - D\left(\Phi_{2}\right)\right\} \text{,}\]
and charges each type a price produced by that type's binding participation constraint. Echoing the basic monopolistic screening model, Theorem \ref{moreconvex} indicates that in the first-best solution type \(\omega_1\) is provided with ``no worse quality;'' that is, \(\Phi_{1,FB}\) is not a strict MPC of \(\Phi_{2,FB}\). In addition, Theorem \ref{moreconvex} tells us that \(t_1 \geq t_2\). Naturally, if there are just two states, \(\omega_1\) is provided with ``higher quality'' than type \(\omega_2\): \(\Phi_{1,FB}\) is an MPS of \(\Phi_{2,FB}\)

In the second-best problem, following standard logic, \(IR_{2}\) and \(IC_{1}\) bind, whereas \(IR_{1}\) and \(IC_{2}\) are slack. The principal's objective is therefore (eliminating constants)
\[\left(1-\rho\right)\left(\frac{1}{1-\rho}\int_{\Delta\left(\Theta\right)}\left(V_2\left(\mu\right) - \rho V_1\left(\mu\right)\right)d\Phi_2\left(\mu\right) - D\left(\Phi_{2}\right)\right) + \rho \left(\int_{\Delta\left(\Theta\right)}V_1\left(\mu\right)d\Phi_1\left(\mu\right) - D\left(\Phi_{1}\right)\right)\text{,}\]
where \(\rho \coloneqq \mathbb{P}\left(\omega_1\right)\). Since 
\[V_2 - \frac{V_2 - \rho V_1}{1-\rho} = \rho\frac{V_1 - V_2}{1-\rho}\] is convex, Theorem \ref{moreconvex} tells us that \(\Phi_{2,SB}\) cannot be a strict MPC of \(\Phi_{2,FB}\) (and for two states, Proposition \ref{twostates} reveals that \(\Phi_{2,SB}\) must be an MPC of \(\Phi_{2,FB}\)). Evidently, \(\Phi_{1,SB} = \Phi_{1,FB}\). As we have argued, all of the standard insights go through:
\begin{remark}
In the ``selling information'' example of this section, in the second-best (screening) solution, there is no output (quality of information) distortion at the top and semi-downward distortion for the ``low'' type relative to the first-best optimum (strictly more information cannot be provided).
\end{remark}

\section{Related Work \& Discussion}\label{section5}

This paper is related to \cite*{curellosinander}, who explore in a single-dimensional setting--either corresponding to a mean-measurable problem with an continuum of states or a binary state--what changes to her indirect payoff lead to greater (or no less) information provision by a persuader. When there are just two states, their first proposition implies Proposition \ref{twostates}. The questions they study, as well as those studied here, are fundamentally comparative statics questions, which connects this work to, e.g., \cite{monotonecomp} and \cite{quahstru}. The example of \(\S\)\ref{sell} is related to \cite*{sinander}, in which the author exploits his novel converse envelope theorem to show that in a information-sales setting, any Blackwell-increasing information allocation is implementable.

Beyond this, economists (and biologists) have been interested in the value of information for decision-makers since \cite*{ram}. This early foray was subsequently followed by the works of Blackwell (\cite*{blackwell, blackwell2}) and \cite*{athey}. The list of other related papers studying the value of information includes \cite*{bergstrom}, who explore the ``fitness value of information'' from a evolutionary perspective; \cite*{lara} who study the value of information using tools from convex analysis; \cite*{radstig}, \cite*{delara}, and \cite*{Chade} who study the marginal valuation of information; and \cite*{azrieli}, who study preference orders over information structures induced by decision problems.

There are other works related on a technical level. \cite*{kleiner2022extreme} study applications of (regular polyhedral) subdivisions to economic settings, with a particular focus on information and mechanism design. \cite{greenosband} are (to my knowledge) the first to connect decision problems to subdivisions. \cite{lambert} shows
that subdivisions are synonymous with elicitability of forecasts. This idea is also present in \cite{frongillo2021general}.

Finally, this paper is also related to the rational inattention literature pioneered by \cite{sims2003implications} and furthered by, e.g., \cite{caplin2022rationally}, \cite{chambers2020costly}, \cite{denti1}, and \cite{denti2022posterior}. \cite{caplinmartin} is especially similar in spirit to this paper. There, they formulate a (binary) relation between joint distributions over actions and states: one such joint distribution dominates another if for every utility function, every experiment consistent with the former is more valuable than every experiment consistent with the latter. In this paper, we construct a binary relation between (equivalence classes) of decision problems--one dominates another if information is more valuable in the former. In this spirit, our paper suggests an easy test for Bayesian rationality: give an experimental participant with some initial endowment of bets an additional refining bet; they should be willing to pay more for information as a result.

\bibliography{sample.bib}

@article{siegel2020economic,
  title={The economic case for probability-based sentencing},
  author={Siegel, Ron and Strulovici, Bruno},
  year={2020},
  journal={Mimeo}
}

@article{weinstein2016effect,
  title={The effect of changes in risk attitude on strategic behavior},
  author={Weinstein, Jonathan},
  journal={Econometrica},
  volume={84},
  number={5},
  pages={1881--1902},
  year={2016}
}

@article{lee201716,
  title={16: Subdivisions and Triangulations of Polytopes},
  author={Lee, Carl W and Santos, Francisco},
  journal={Handbook of Discrete and Computational Geometry. Ed. by Jacob E. Goodman, Joseph O’Rourke, and Csaba D. T{\'o}th},
  pages={415--447},
  year={2017}
}

@article{yoder2022designing,
  title={Designing incentives for heterogeneous researchers},
  author={Yoder, Nathan},
  journal={Journal of Political Economy},
  volume={130},
  number={8},
  pages={2018--2054},
  year={2022}}

@article{denti2022posterior,
  title={Posterior separable cost of information},
  author={Denti, Tommaso},
  journal={American Economic Review},
  volume={112},
  number={10},
  pages={3215--3259},
  year={2022}}

@article{monotonecomp,
author = {Milgrom, Paul and Shannon, Chris},
title = {Monotone comparative statics},
journal = {Econometrica},
volume = {62},
number = {1},
pages = {157-180},
year = {1994}
}

@article{quahstru,
 author = {John K.-H. Quah and Bruno Strulovici},
 journal = {Econometrica},
 number = {6},
 pages = {1949--1992},
 title = {Comparative Statics, Informativeness, and the Interval Dominance Order},
 volume = {77},
 year = {2009}
}

@article{battigalli2016note,
  title={A note on comparative ambiguity aversion and justifiability},
author={Battigalli, Pierpaolo and Cerreia-Vioglio, Simone and Maccheroni, Fabio and Marinacci, Massimo},
  journal={Econometrica},
  volume={84},
  number={5},
  pages={1903--1916},
  year={2016}}

@article{safety,
title = {Safety, in Numbers},
journal = {Mimeo},
year = {2023},
month = {Oct},
author = {Pease, Marilyn and Whitmeyer, Mark}}

@article{kleiner2022extreme,
  title={Applications of Power Diagrams to Mechanism and Information Design},
  author={Kleiner, Andreas and Moldovanu, Benny and Strack, Philipp and Whitmeyer, Mark},
  journal={Mimeo},
  year={2023}
}

@article{chambers2020costly,
  title={Costly information acquisition},
  author={Chambers, Christopher P and Liu, Ce and Rehbeck, John},
  journal={Journal of Economic Theory},
  volume={186},
  pages={104979},
  year={2020}}

@article{greenosband,
    author = {Green, Edward J. and Osband, Kent},
    title = {A Revealed Preference Theory for Expected Utility},
    journal = {The Review of Economic Studies},
    volume = {58},
    number = {4},
    pages = {677-695},
    year = {1991},
    month = {06}}

@article{frongillo2021general,
  title={General truthfulness characterizations via convex analysis},
  author={Frongillo, Rafael M and Kash, Ian A},
  journal={Games and Economic Behavior},
  volume={130},
  pages={636--662},
  year={2021}
}

@article{lambert,
    author = {Lambert, Nicholas S.},
    title = {Elicitation and Evaluation of Statistical Forecasts},
    journal = {Mimeo},
    year = {2019},
    month = {06}}

@book{rockafellar2015convex,
  title={Convex Analysis},
  author={Rockafellar, Ralph Tyrell},
  year={1970},
  publisher={Princeton University Press}
}

@inproceedings{blackwell,
address = "Berkeley, Calif.",
author = "Blackwell, David",
booktitle = "Proceedings of the Second Berkeley Symposium on Mathematical Statistics and Probability",
pages = "93--102",
publisher = "University of California Press",
title = "Comparison of Experiments",
year = "1951"
}

@inproceedings{radstig,
Author = {Radner, Roy and Stiglitz, Joseph},
Year = {1984},
booktitle = {M. Boyer, R. Kihlstrom (Eds.), Bayesian Models of Economic Theory},
Title = {A nonconcavity in the value of information}, 
Publisher = {Elsevier, Amsterdam},
pages = {33-52}}

@article{Chade,
title = {Another Look at the {Radner–Stiglitz} Nonconcavity in the Value of Information},
journal = {Journal of Economic Theory},
volume = {107},
number = {2},
pages = {421-452},
year = {2002},
author = {Hector Chade and Edward Schlee}
}

@article{delara,
title = {A tight sufficient condition for {Radner–Stiglitz} nonconcavity in the value of information},
journal = {Journal of Economic Theory},
volume = {137},
number = {1},
pages = {696-708},
year = {2007},
author = {De Lara, Michel and Gilotte, Laurent}}

@article{azrieli,
title = {The value of a stochastic information structure},
journal = {Games and Economic Behavior},
volume = {63},
number = {2},
pages = {679-693},
year = {2008},
author = {Yaron Azrieli and Ehud Lehrer}}

@article{blackwell2,
 author = {David Blackwell},
 journal = {The Annals of Mathematical Statistics},
 number = {2},
 pages = {265--272},
 title = {Equivalent Comparisons of Experiments},
 volume = {24},
 year = {1953}
}

@article{bergstrom,
author = {Donaldson-Matasci, Matina C. and Bergstrom, Carl T. and Lachmann, Michael},
title = {The fitness value of information},
journal = {Oikos},
volume = {119},
number = {2},
pages = {219-230},
year = {2010}
}

@article{lara,
author = {De Lara, Michel and Gossner, Olivier},
title = {Payoffs-Beliefs Duality and the Value of Information},
journal = {SIAM Journal on Optimization},
volume = {30},
number = {1},
pages = {464-489},
year = {2020}
}

@article{athey,
title = {The value of information in monotone decision problems},
journal = {Research in Economics},
volume = {72},
number = {1},
pages = {101-116},
year = {2018},
author = {Susan Athey and Jonathan Levin}}

@article{ram,
 author = {F. P. Ramsey},
 journal = {The British Journal for the Philosophy of Science},
 number = {1},
 pages = {1--4},
 title = {Weight or the Value of Knowledge},
 volume = {41},
 year = {1990}
}

@article{clearadvice,
    author = {Szalay, Dezs\"{o}},
    title = {The Economics of Clear Advice and Extreme Options},
    journal = {The Review of Economic Studies},
    volume = {72},
    number = {4},
    pages = {1173-1198},
    year = {2005},
    month = {10}
}

@article{curellosinander,
author = {Curello, Gregorio and Sinander, Ludvig},
title = {The Comparative Statics of Persuasion},
journal = {Mimeo},
year = {2022}}

@article{sinander,
title={The converse envelope theorem},
  author={Sinander, Ludvig},
  journal={Econometrica},
  volume={90},
  number={6},
  pages={2795--2819},
  year={2022}}

@article{sims2003implications,
  title={Implications of rational inattention},
  author={Sims, Christopher A},
  journal={Journal of Monetary Economics},
  volume={50},
  number={3},
  pages={665--690},
  year={2003}
}

@article{caplin2022rationally,
  title={Rationally Inattentive Behavior: Characterizing and Generalizing Shannon Entropy},
  author={Caplin, Andrew and Dean, Mark and Leahy, John},
  journal={Journal of Political Economy},
  year={2022},
  pages = {1676--1715},
  number = {6},
  volume = {130},
}

@article{caplinmartin,
  title={Comparison of Decisions under Unknown Experiments},
  author={Caplin, Andrew and Martin, Daniel},
  journal={Journal of Political Economy},
  year={2021},
  pages = {3185--3205},
  number = {11},
  volume = {129},
}

@article{denti1,
Author = {Denti, Tommaso and Marinacci, Massimo and Rustichini, Aldo},
Title = {Experimental Cost of Information},
Journal = {American Economic Review},
Volume = {112},
Number = {9},
Year = {2022},
Month = {September},
Pages = {3106-23}}

@article{kam,
 author = {Emir Kamenica and Matthew Gentzkow},
 journal = {The American Economic Review},
 number = {6},
 pages = {2590-2615},
 title = {Bayesian Persuasion},
 volume = {101},
 year = {2011}
}

\appendix

\section{Omitted Proofs}\label{theappendix}

\subsection{Lemma \ref{finernecessitylemma}}\label{finernecessitylemmaproof}
\begin{proof}
Suppose for the sake of contraposition \(\hat{C} \not\succeq C\). This implies for some \(\hat{C}_j \in \hat{C}\), there is no \(C_i\) such that \(\hat{C}_j \succeq C_i\). This means there exist \(\mu, \mu' \in \hat{C}_j\) and some \(C_i \in C\) with \(\mu \in C_i\) and \(\mu' \notin C_i\). By definition \(\hat{V}\) is affine on \(\hat{C}_j\), so for all \(\lambda \in \left[0,1\right]\),
\[\lambda \hat{V}\left(\mu\right) + \left(1-\lambda\right)\hat{V}\left(\mu'\right) = \hat{V}\left(\lambda \textcolor{OrangeRed}{\mu} + \left(1-\lambda\right) \mu'\right)\text{.}\]
By construction, for all \(\lambda \in \left(0,1\right)\),
\[\lambda V\left(\mu\right) + \left(1-\lambda\right)V\left(\mu'\right) > V\left(\lambda \textcolor{OrangeRed}{\mu} + \left(1-\lambda\right) \mu'\right)\text{.}\]
Combining these, and maintaining the convention \(\hat{W} = \hat{V} - V\), we obtain that 
\[\lambda \hat{W}\left(\mu\right) + \left(1-\lambda\right)\hat{W}\left(\mu'\right) < \hat{W}\left(\lambda \textcolor{OrangeRed}{\mu} + \left(1-\lambda\right) \mu'\right)\text{,}\]
so \(\hat{V}-V\) is not convex. \end{proof}

\subsection{Lemma \ref{newendoginfo} Proof}\label{newendoginfoproof}
\begin{proof}
    Suppose for the sake of contraposition that \(\hat{V}-V\) is not convex on \(\Delta \left(\Theta\right)\). As \(\hat{V} - V\), being the difference of two continuous functions, is continuous, it is not convex on \(\inter \Delta\left(\Theta\right)\).
    
    Let \(\rho\left(\mu\right)\) be some strictly convex function on \(\Delta\left(\Theta\right)\) and for an arbitrary \(\varepsilon > 0\) define function \(c_{\varepsilon}\left(\mu\right) \coloneqq \varepsilon \rho\left(\mu\right) + \hat{V}\left(\mu\right)\). By construction, for all \(\varepsilon > 0\), \(c_{\varepsilon}\) is strictly convex.

    Moreover, \(\hat{V} - c_{\varepsilon} = - \varepsilon \rho\) is strictly concave for all \(\varepsilon > 0\), so in the agent's flexible information acquisition problem for the transformed decision problem, the unique solution is for her to acquire the degenerate distribution on her prior, \(\delta_{\mu_{0}}\), for any prior \(\mu_{0} \in \inter\Delta \left(\Theta\right)\).
    
    In contrast, in the initial decision problem, \(\mathcal{D}\), the agent's objective in her flexible information acquisition problem is 
    \[V - c_{\varepsilon} = V - \hat{V} - \varepsilon \rho \text{.}\]
    As \(W \coloneqq V - \hat{V}\) is not concave on \(\inter \Delta\left(\Theta\right)\), for all sufficiently small \(\varepsilon > 0\), 
    \[\tag{\(A1\)}\label{eq1}\lambda \left(W\left(\mu\right) - \varepsilon \rho\left(\mu\right)\right) + \left(1-\lambda\right) \left(W(\mu') - \varepsilon \rho\left(\mu'\right)\right) >  W\left(\lambda \mu + \left(1-\lambda\right)\mu'\right) - \varepsilon \rho\left(\lambda \mu + \left(1-\lambda\right)\mu'\right)\]
    for some \(\lambda \in \left(0,1\right)\) and \(\mu, \mu' \in \inter \Delta\left(\Theta\right)\). Accordingly, as we may set \(\mu_{0} = \lambda \mu + \left(1-\lambda\right)\mu'\), there exists a \(\mu_{0} \in \inter \Delta \Theta\) such that the agent strictly prefers acquiring some information to learning nothing. That is, for any optimal \(\Phi^*\), \(\hat{\Phi}^* = \delta_{\mu_{0}}\) is a strict MPC of \(\Phi^*\).
\end{proof}

\subsection{Lemma \ref{newexoginfo} Proof}\label{newexoginfoproof}
\begin{proof}
    Suppose again for the sake of contraposition that \(\hat{V}-V\) is not convex, without loss of generality on \(\inter \Delta \left(\Theta\right)\). Inequality \ref{eq1} implies there exist \(\lambda \in \left(0,1\right)\), \(\mu, \mu' \in \inter \Delta\left(\Theta\right)\), and \(\mu_{0}^{*} = \lambda \mu + \left(1-\lambda\right) \mu'\) such that
    \[\lambda \left(V\left(\mu\right) - \hat{V}\left(\mu\right)\right) + \left(1-\lambda\right) \left(V(\mu') - \hat{V}(\mu')\right) >  V\left(\mu_{0}^{*}\right) - \hat{V}\left(\mu_{0}^{*}\right)\text{,}\]
    which holds if and only if 
    \[\mathbb{E}_\Phi V\left(\mu\right) - V\left(\mu_{0}^{*}\right) > \mathbb{E}_\Phi \hat{V}\left(\mu\right) - \hat{V}\left(\mu_{0}^{*}\right)\text{,}\]
    where \(\Phi\) is the binary Bayes-plausible distribution with support \(\left\{\mu,\mu'\right\}\).
\end{proof}

\subsection{Proposition \ref{nomoreprop} Proof}\label{nomorepropproof}
The proof builds off of the following lemma, which may be of independent interest. Let \(C\) be the subdivision corresponding to \(\mathcal{D}\), and let \(\Phi\) be a distribution of posteriors whose (finite) support is such that every point \(\mu_i \in \supp \Phi\) lies in the relative interior of a cell \(C_i \in C\) and each \(\mu_i\) is in a distinct \(C_i\). We call such a \(\Phi\) \emph{Non-Redundant}, as it does not contain multiple posteriors that justify the same action.

We say that a UPS cost \(D\) \emph{Generates Learning} \(\Phi\) if it is optimal for an agent with cost \(D\) to acquire \(\Phi\) and \(\Phi\) is uniquely optimal in the sense that any optimally acquired distribution over posteriors is supported on a subset of \(\supp \Phi\). Then,
\begin{lemma}\label{nonredundancy}
    If \(\Phi\) is non-redundant, there exists a cost \(D\) that generates it.
\end{lemma}
\begin{proof}
    The proof and discussion of this lemma lie in the supplementary appendix. \end{proof}
Now let us prove the proposition.
\begin{proof}[Proof of Proposition \ref{nomoreprop}.]
    \(\left(\Leftarrow\right)\) If \(V\) is affine, \(V - c\) is strictly concave, so it is uniquely optimal to acquire no information in Problem \ref{flex}. Consequently, any solution to Problem \ref{flex} is an MPC of the solution to Problem \ref{flexhat}. If \(\hat{V} - V\) is affine, the solutions to Problems \ref{flex} and \ref{flexhat} are identical.

    \bigskip

    \noindent \(\left(\Rightarrow\right)\) We prove the result by contraposition. Suppose neither \(V\) nor \(\hat{V} - V\) is affine. If \(\hat{V} - V\) is not convex, by Theorem \ref{moreconvex}, we are done, so let \(\hat{V}-V \eqqcolon \hat{W}\) be convex. As neither \(\hat{W}\) nor \(V\) is affine and as we can add affine functions to value functions without altering decision problems, we specify without loss of generality that the sets \(Z \coloneqq \left\{\mu \in \Delta \left(\Theta\right) \colon \hat{W}\left(\mu\right) \leq 0\right\}\) and \(Y \coloneqq \left\{\mu \in \Delta \left(\Theta\right) \colon \hat{W}\left(\mu\right) > 0\right\}\) are such that
    \begin{enumerate}[label={(\roman*)},noitemsep,topsep=0pt]
        \item \(Z\) is a polytope of full dimension in \(\mathbb{R}^{n-1}\); 
        \item Neither \(V\) nor \(\hat{W}\) is affine on \(Z\)
        \item \(Y\) is of full dimension in \(\mathbb{R}^{n-1}\); and
        \item There exist two points \(\mu_1, \mu_{2} \in \inter Y\) and \(\lambda \in \left(0,1\right)\) for which \(\mu_{0} = \lambda \mu_{1} + \left(1-\lambda\right) \mu_{2} \in \inter Z\) and \(\hat{V}\) is locally affine around both \(\mu_1\) and \(\mu_2\).\footnote{We discuss this imposition further in the supplementary appendix.}
    \end{enumerate}
    
    Define \(\tilde{W}\left(\mu\right) \coloneqq \max\left\{V\left(\mu\right), \hat{V}\left(\mu\right)\right\}\) and let \(\tilde{C}\) be the corresponding subdivision. By construction, there are at least two cells \(\tilde{C}_1, \tilde{C}_2 \subseteq Z\). Take \(\mu_1, \mu_{2}\) and \(\mu_{0}\) as previously specified. By construction, we can find \(\tilde{\mu}_1 \in \inter \tilde{C}_1\), \(\tilde{\mu}_2 \in \inter \tilde{C}_2\) and \(\tilde{\lambda} \in \left(0,1\right)\) such that \(\mu_{0} = \tilde{\lambda} \tilde{\mu}_1 + \left(1-\tilde{\lambda}\right) \tilde{\mu}_2\) and \(\tilde{\mu}_1\) and \(\tilde{\mu}_2\) do not lie on the line segment between \(\mu_1\) and \(\mu_2\).

    Let \(\Phi\) be a distribution with support \(\left\{\mu_{1},\mu_{2}, \tilde{\mu}_1,\tilde{\mu}_2\right\}\). By construction, \(\Phi\) is non-redundant when the agent's value function is \(\tilde{W}\). Thus, Lemma \ref{nonredundancy} implies there is a cost function, \(D\), that generates \(\Phi\). Moreover, for prior \(\mu_{0}\), \(D\) renders the Bayes-plausible binary distribution with support \(\left\{\mu_{1},\mu_{2}\right\}\) uniquely optimal when the agent's value function is \(\hat{V}\) and the Bayes-plausible binary distribution with support \(\left\{\tilde{\mu}_1,\tilde{\mu}_2\right\}\) uniquely optimal when the agent's value function is \(V\). These two distributions are Blackwell-incomparable. \end{proof}

\subsection{Lemma \ref{convexandfiner} Proof}\label{convexandfinerproof}
\begin{proof}
    Suppose \(\hat{C} \succeq C\). That is, for every \(\hat{C}_j \in \hat{C}\), there exists a \(C_i \in C\) such that \(\hat{C}_j \subseteq C_i\). Define \[\hat{W}\left(\mu\right) \coloneqq \hat{V}\left(\mu\right) - V\left(\mu\right) = \max\left\{\hat{\alpha} \cdot \mu + \hat{\beta}, V\left(\mu\right)\right\} - V\left(\mu\right)\text{,}\] where \(\hat{\alpha} \cdot \mu + \hat{\beta}\) (\(\hat{\beta} \in \mathbb{R}\), \(\hat{\alpha} \in \mathbb{R}^{n-1}\)) is the expected payoff from taking the new action \(\hat{a}\). By assumption, \[\hat{C}_{j} \coloneqq \left\{\mu \in \Delta \colon \hat{\alpha} \cdot \mu + \hat{\beta} \geq V\left(\mu\right)\right\} \subseteq C_{j}\] for some \(j \in \left\{1,\dots,m\right\}\). Moreover, by the continuity of \(V\) and \(\hat{V}\), for all \(\mu \in \cup_{i \neq j} C_i\), \(\hat{W}\left(\mu\right) = V\left(\mu\right) - V\left(\mu\right) = 0\).
    
    For all \(\mu \in C_j\), \(V\left(\mu\right) = \alpha \cdot \mu + \beta\) (\(\beta \in \mathbb{R}\), \(\alpha \in \mathbb{R}^{n-1}\)). Accordingly, for all \(\mu \in C_j\), 
    \[\hat{W}\left(\mu\right) = \max\left\{\hat{\alpha} \cdot \mu + \hat{\beta}, \alpha \cdot \mu + \beta\right\} - \left(\alpha \cdot \mu + \beta\right) = \max\left\{0,\left(\hat{\alpha} - \alpha\right) \cdot \mu + \hat{\beta} - \beta\right\}\text{,}\]
    which means that for all \(\mu \in \Delta\left(\Theta\right)\), \[\hat{W}\left(\mu\right) = \max\left\{0,\left(\hat{\alpha} - \alpha\right) \cdot \mu + \hat{\beta} - \beta\right\}\text{,}\]
    which is convex, being the maximum of two affine functions. \end{proof}

\subsection{Proposition \ref{generic} Proof}\label{genericproof}
\begin{proof} \(\left(\Rightarrow\right)\) This direction mostly repeats the proof of Lemma \ref{seqextremal}. Fix an arbitrary \(b \in B\) and let \(\hat{V}_b^u\) denote the agent's value function when the set of actions is \(A \cup \left\{b\right\}\) and the utilities from the new actions are \(u \in \mathbb{R}^{B \times \Theta}\). Since \(b\) is either strictly refining or strictly dominated, \(\hat{V}_b^u - V\) is convex. Moreover, if \(b\) is strictly dominated, \(\mathbb{E}_{\mu} u\left(b,\theta\right) < V\left(\mu\right)\) for all \(\mu \in \Delta \left(\Theta\right)\). Consequently, there exists an open ball around \(u\) such that for all \(\tilde{u}\) in that open ball, \(\mathbb{E}_{\mu} u\left(b,\theta\right) < V\left(\mu\right)\) for all \(\mu \in \Delta \left(\Theta\right)\). Now let \(b\) be strictly refining. Let \(C_i\) be the cell of \(C\) for which the following inclusion holds:
\[\left\{\mu \in \Delta\left(\Theta\right) \ \vert \ \mathbb{E}_{\mu}u\left(\hat{a},\theta\right) \geq V\left(\mu\right)\right\} \subseteq C_i\text{.}\]
By the definition of strictly refining, for all \(\mu \notin C_i\), \(\mathbb{E}_{\mu} u\left(b,\theta\right) < V\left(\mu\right)\). Accordingly, there exists an open ball around \(u\) such that for all \(\tilde{u}\) in that open ball, \(\mathbb{E}_{\mu} u\left(b,\theta\right) < V\left(\mu\right)\) for all \(\mu \in \Delta \left(\Theta\right) \setminus C_i\). Combining these two observations, we see that for all \(\tilde{u}\) in some open ball around \(u\), \(b\) is strictly refining, so \(\hat{V}_b^{\tilde{u}} - V\) is convex. Finally, \(\hat{V}^{\tilde{u}} - V = \max\left(V_b^{\tilde{u}}\right)_{b \in B} - V = \max\left\{\left(V_b^{\tilde{u}} - V\right)_{b \in B}\right\}\) is convex, being the maximum of convex functions.

\bigskip

    \noindent \(\left(\Leftarrow\right)\) Suppose for the sake of contraposition that \(B\) is not totally strictly refining. That is, there exists a \(b \in B\), an \(\mu \in \Delta\), and two undominated (in \(A\)) actions \(a_1, a_2 \in A\) for which \[\mathbb{E}_{\mu}u\left(b,\theta\right) \geq \mathbb{E}_{\mu}u\left(a_1,\theta\right) = \mathbb{E}_{\mu}u\left(a_2,\theta\right) = V\left(\mu\right)\text{.}\]
    Without loss of generality, we may assume \(\hat{V}\left(\mu\right) = \mathbb{E}_{\mu}u\left(b,\theta\right)\) as any \(b' \neq b\) for which 
    \(\mathbb{E}_{\mu}u\left(b',\theta\right) = \hat{V}\left(\mu\right) > \mathbb{E}_{\mu}u\left(b,\theta\right)\) is neither dominated nor strictly refining itself, so we could just replace \(b\) in the proof with \(b'\).
    
    Now, pick some \(\theta' \in \Theta\) that occurs with positive probability under \(\mu\). Define \(\tilde{u}\left(b,\theta\right) = u\left(b,\theta\right)\) for all \(\theta \neq \theta'\) and \(\tilde{u}\left(b,\theta'\right) = u\left(b,\theta'\right) + \varepsilon\). Then, for all \(\varepsilon > 0\),
\[\label{eq2}\tag{\(A2\)}\mathbb{E}_{\mu}\tilde{u}\left(b,\theta\right) > \hat{V}\left(\mu\right) \geq \mathbb{E}_{\mu}u\left(a_1,\theta\right) = \mathbb{E}_{\mu}u\left(a_2,\theta\right) = V\left(\mu\right)\text{.}\]
By the continuity of each of the four functions, \(\mathbb{E}_{\mu}\tilde{u}\left(b,\theta\right)\), \(\hat{V}\left(\mu\right)\), \(\mathbb{E}_{\mu}u\left(a_1,\theta\right)\), and \(\mathbb{E}_{\mu}u\left(a_2,\theta\right)\) in \(\mu\), for all \(\mu'\) in some open ball (understanding \(\Delta\left(\Theta\right)\) as a subset of \(\mathbb{R}^{n-1}\), equipped with the Euclidean metric) around \(\mu\),
\[\mathbb{E}_{\mu'}\tilde{u}\left(b,\theta\right) > \max\left\{\hat{V}\left(\mu'\right), \mathbb{E}_{\mu'}u\left(a_1,\theta\right), \mathbb{E}_{\mu'}u\left(a_2,\theta\right), V\left(\mu'\right)\right\}\text{.}\]
Moreover, neither \(a_1\) nor \(a_2\) is weakly dominated so for any open ball \(B_{\eta}\left(\mu\right)\) around \(\mu\), there exist \(\mu_1, \mu_{2} \in B_{\eta}\left(\mu\right)\) such that \(\mathbb{E}_{\mu_{1}}u\left(a_1,\theta\right) > \max_{a \in A \setminus \left\{a_1\right\}} \mathbb{E}_{\mu_{1}}u\left(a,\theta\right)\) and \(\mathbb{E}_{\mu_{2}}u\left(a_2,\theta\right) > \max_{a \in A \setminus \left\{a_2\right\}} \mathbb{E}_{\mu_{2}}u\left(a,\theta\right)\).

This implies there is no \(C_i \in C\) such that 
\[\left\{\mu \in \Delta\left(\Theta\right) \ \vert \ \mathbb{E}_{\mu}\tilde{u}\left(b,\theta\right) \geq \max_{a \in A \cup B}\mathbb{E}_{\mu}u\left(a,\theta\right)\right\} \subseteq C_i\text{,}\]
so \(C \not\succeq \hat{C}^{\tilde{u}}\), where \(\hat{C}^{\tilde{u}}\) denotes the subdivision in the perturbed transformed decision problem. The contraposition of Lemma \ref{finernecessitylemma}, therefore, produces the result. \end{proof}

\subsection{Proposition \ref{lessflexprop} Proof}\label{lessflexpropproof}
\begin{proof}
    Suppose for the sake of contraposition there are leftovers and the elimination is not inconsequential. Namely, there is some \(a_i \in \hat{A}\) for which there exists an \(\mu \in \Delta \left(\Theta\right)\) such that \[\mathbb{E}_{\mu}u\left(a_i,\theta\right) > \max_{a \in A \setminus \left\{a_i\right\}}\mathbb{E}_{\mu}u\left(a,\theta\right)\text{;}\]
    and there is some \(a_j \in A \setminus \hat{A}\) for which there exists an \(\mu \in \Delta \left(\Theta\right)\) such that \[\mathbb{E}_{\mu}u\left(a_j,\theta\right) > \max_{a \in A \setminus \left\{a_j\right\}}\mathbb{E}_{\mu}u\left(a,\theta\right)\text{.}\]
    
    By construction, for all \(\mu \in C_i\), \(\mathbb{E}_{\mu} u\left(a_i,\theta\right) = V\left(\mu\right) = \hat{V}\left(\mu\right)\) and for all \(\mu \in \inter C_i\), \[\mathbb{E}_{\mu} u\left(a_i,\theta\right) = V\left(\mu\right) = \hat{V}\left(\mu\right) > \max_{a \in A\setminus\left\{a_i\right\}}\mathbb{E}_{\mu} u\left(a,\theta\right) \geq \max_{a \in \hat{A}\setminus\left\{a_i\right\}}\mathbb{E}_{\mu} u\left(a,\theta\right)\text{.}\]
    Let \(a_j \in A \setminus \hat{A}\); namely, \(a_j\) is one of the actions taken away. For all \(\mu \in \inter C_j\), 
    \[\mathbb{E}_{\mu} u\left(a_j,\theta\right) = V\left(\mu\right) > \hat{V}\left(\mu\right)\text{.}\]

    By construction, there exist points \(\mu', \mu_{0} \in \inter C_i\), \(\mu \in \inter C_j\), and weight \(\lambda \in \left(0,1\right)\) such that \(\mu_{0} = \lambda \mu + \left(1-\lambda\right) \mu'\), and
    \[\lambda \underbrace{\left(\hat{V}\left(\mu\right)-V\left(\mu\right)\right)}_{ < 0} + \left(1-\lambda\right)\underbrace{\left(\hat{V}(\mu')-V(\mu')\right)}_{= 0} - \underbrace{\left(\hat{V}\left(\mu_{0}\right)-V\left(\mu_{0}\right)\right)}_{= 0} < 0 \text{,}\]
    so \(\hat{V}-V\) is not convex.
\end{proof}

\subsection{Proposition \ref{affinetrans} Proof}\label{affinetransproof}
\begin{proof}
If \(k = 1\), the result is immediate, so let \(k \neq 1\). For any pair of points \(\mu_1 \neq \mu_{2} \in \Delta\left(\Theta\right)\) with corresponding optimal actions \(a_1 \neq a_2\), \(\lambda \in \left[0,1\right]\), and optimal action at \(\mu^\dagger \coloneqq \lambda \mu_{1} + \left(1-\lambda\right) \mu_{2}\), \(a^{\dagger}\),
\[\begin{split}
    \lambda \hat{W}\left(\mu_1\right) + \left(1-\lambda\right)\hat{W}\left(\mu_2\right) &> \hat{W}\left(\mu^{\dagger}\right) \quad \Leftrightarrow\\
    \left(k-1\right)\left[\lambda \mathbb{E}_{\mu_{1}}u\left(a_1,\theta\right) + \left(1-\lambda\right) \mathbb{E}_{\mu_{2}}u\left(a_2,\theta\right)\right] &> \left(k-1\right)\mathbb{E}_{\mu^{\dagger}}u\left(a^{\dagger},\theta\right) \quad \Leftrightarrow\\
    k &> 1 \text{,}
\end{split}\]
recalling that \(\hat{W}\left(\mu\right) \coloneqq \hat{V}\left(\mu\right)-V\left(\mu\right)\). \end{proof}

\subsection{Proposition \ref{riskaverseprop} Proof}
\begin{proof}
    Recall that \(A\) is finite. We need only establish the result when we make the agent more risk averse, as the proof for when we make the agent more risk loving is virtually identical \textit{mutatis mutandis}.\footnote{In fact, it is easier to tackle the scenario in which the agent becomes more risk loving, due to the fact that previously dominated actions cannot become undominated when an agent's love for risk increases (Proposition 1 in both \cite*{battigalli2016note} and \cite{weinstein2016effect}).} Understanding the agent's utility \(u\) to be an element of the Euclidean space \(\mathbb{R}^{A \times \Theta}\), we will show that for any strictly increasing, strictly concave \(\phi\), if \(\hat{C} \succeq C\), then for any sufficiently small \(\varepsilon > 0\), there exists a perturbed utility \(u^{\dagger} \in B_{\varepsilon}(u)\) for which \(\hat{C}^{\dagger} \not\succeq C^{\dagger}\), where \(\hat{C}^{\dagger}\) and \(C^{\dagger}\) are the perturbed subdivisions in the transformed and initial decision problems, respectively.

    We may do the following three things without loss of generality. First, we specify that there are just two states, \(\Theta = \left\{0,1\right\}\). This is because when there are more than two states, we could just specialize to an arbitrary edge of \(\Delta\left(\Theta\right)\). Second, we assume that there are no ``duplicate actions,'' i.e., that there do not exist \(a, a' \in A\) for which \(u(a,\theta) = u(a',\theta)\) for all \(\theta \in \Theta\). Third, we impose that in both the initial and transformed decision problems, at any belief \(\mu \in \left[0,1\right]\), there are at most two actions optimal at this belief. That is, there are no weakly dominated actions that are not strictly dominated. Genericity facilitates this: reducing the payoffs in both states of any weakly dominated action by any \(\varepsilon > 0\) makes it strictly dominated.
    
    The subdivision in the initial decision problem is \[C = \left\{C_1,\dots,C_m\right\}\text{,}\]
    labeled so that
    \[C_1 = \left[0,\mu_1\right], C_2 = \left[\mu_1,\mu_2\right], \dots , C_m = \left[\mu_m,1\right]\text{,}\]
    where
    \[0 < \mu_1 < \cdots < \mu_m < 1\text{.}\]
    
    For \(i \in \left\{1,2\right\}\), let \(\alpha_i \coloneqq u(a_i,0)\) and \(\beta_i \coloneqq u(a_i,1)\), where, by construction, \(\alpha_1 > \alpha_2\) and \(\beta_2 > \beta_1\). The belief at which the agent is indifferent between the two actions in the initial decision problem is
    \[\mu_1 = \frac{\alpha_1 - \alpha_2}{\alpha_1 - \alpha_2 + \beta_2 - \beta_1}\text{.}\]
    By Proposition 1 in \cite{battigalli2016note} and Proposition 1 in \cite{weinstein2016effect}, no action in the set of actions that are not strictly dominated in the initial decision problem can become strictly dominated when the agent becomes more risk averse. Note further that the strict monotonicity of \(\phi\) ensures that action \(a_1\) must remain optimal for some interval of beliefs \(\left[0,\hat{\mu}_1\right]\) where \(\hat{\mu}_1 > 0\); that is, the ``first element'' of subdivision \(\hat{C}\), going left-to-right along \(\left[0,1\right]\), must still correspond to action \(a_1\).

    Consequently, a transformation that results in a weakly finer subdivision must either be such that
    \begin{enumerate}[noitemsep,topsep=0pt]
        \item \textbf{Case 1.} \(\mu_1 = \hat{\mu}_1\); or
        \item \textbf{Case 2.} \(\mu_1 > \hat{\mu}_1\).
    \end{enumerate}
    In the supplemental appendix, we show that in either case, the preservation or refinement of the subdivision is non-generic.\end{proof}

\subsection{Theorem \ref{mugenerate} Proof}\label{mugenerateproof}
\begin{proof}
    \(\left(\Leftarrow\right)\) Suppose \(\hat{V}\) shift-majorizes \(V\) and fix \(\mu_0 \in \inter \Delta\left(\Theta\right)\). Then for any \(F \in \mathcal{F}_{\mu_{0}}\),
\[\begin{split}
    \mathbb{E}_{F} \hat{V} - \hat{V}\left(\mu_{0}\right) &\geq \mathbb{E}_{F} \left[V - \ell\right] - \hat{V}\left(\mu_{0}\right)\\
    &= \mathbb{E}_F V - \ell\left(\mu_{0}\right)- \hat{V}\left(\mu_{0}\right)\\
    &= \mathbb{E}_F V - V\left(\mu_{0}\right)\text{.}
\end{split}\]

\medskip

\noindent \(\left(\Rightarrow\right)\) Fix \(\mu_0 \in \inter \Delta\left(\Theta\right)\) and suppose that for any \(F \in \mathcal{F}_{\mu_{0}}\)
\[\mathbb{E}_{F} \left[\hat{V} - V\right] \geq \hat{V}\left(\mu_{0}\right) - V\left(\mu_{0}\right)\text{.}\]
Recall our convention that \(\hat{W} \coloneqq \hat{V}-V\), and let \(\breve{W}\) be the lower convex envelope of \(\hat{W}\):
\[\breve{W}\left(\mu\right) \coloneqq \sup\left\{\left.g\left(\mu\right) \right| g \text{ is convex and } g\left(\mu\right) \leq \hat{W}\left(\mu\right) \text{ for all } \mu \in \Delta\left(\Theta\right)\right\}\text{.}\]
As the supremum of convex functions, \(\breve{W}\) is convex; and by construction \(\breve{W}\) lies below \(\hat{W}\) pointwise, which implies that \(\breve{W}\left(\mu_{0}\right) \leq \hat{W}\left(\mu_{0}\right)\).
\begin{claim}\label{equalclaim}
    \(\breve{W}\left(\mu_{0}\right) = \hat{W}\left(\mu_{0}\right)\).
\end{claim}
\begin{proof}
    Suppose for the sake of contradiction that \(\breve{W}\left(\mu_{0}\right) < \hat{W}\left(\mu_{0}\right)\). As \(\hat{W}\) is continuous on \(\Delta\left(\Theta\right)\), Corollary 2 from \cite{kam} implies that there exists \(F \in \mathcal{F}_{\mu_{0}}\) for which
    \[\mathbb{E}_F \hat{W} < \hat{W}\left(\mu_{0}\right)\text{,}\]
    a contradiction.\end{proof}
\(\left(\mu_0, \breve{W}\left(\mu_{0}\right)\right)\) lies on the boundary of \(\epi \breve{W}\), which is a convex set. Thus, by the supporting hyperplane theorem there exists a supporting hyperplane \(\ell\) such that \(\breve{W}\left(\mu\right) \geq \ell\left(\mu\right)\) for all \(\mu \in \Delta\left(\Theta\right)\) and \(\breve{W}\left(\mu_{0}\right) = \ell\left(\mu_{0}\right)\). Finally, pointwise dominance of \(\hat{W}\) over \(\breve{W}\) implies \(\hat{W}\left(\mu\right) \geq \breve{W}\left(\mu\right) \geq \ell\left(\mu\right)\) for all \(\mu \in \Delta\left(\Theta\right)\), and Claim \ref{equalclaim} implies \(\hat{W}\left(\mu_{0}\right) = \breve{W}\left(\mu_{0}\right) = \ell\left(\mu_{0}\right)\).
\end{proof}

\subsection{Theorem \ref{muless} Proof}\label{mulessproof}
\begin{proof}
    Suppose for the sake of contraposition that \(\hat{V}\) does not shift-majorize \(V\). Take \(\hat{W} \coloneqq \hat{V} - V\) and recall that \(\breve{W}\) is the lower convex envelope of \(\hat{W}\). As \(\hat{V}\) does not shift-majorize \(V\), \(\breve{W}(\mu_0) < \hat{W}(\mu_0)\). This means that there exists a Bayes-plausible distribution over posteriors, \(\Phi^{\dagger}\), such that \(\mathbb{E}_{\Phi^{\dagger}} \hat{W} < \hat{W}(\mu_0)\).

    Let \(\rho\left(\mu\right)\) be some strictly convex function on \(\Delta\left(\Theta\right)\) and for an arbitrary \(\varepsilon > 0\) define function \(c_{\varepsilon}\left(\mu\right) \coloneqq \varepsilon \rho\left(\mu\right) + \hat{V}\left(\mu\right)\). By construction, for all \(\varepsilon > 0\), \(c_{\varepsilon}\) is strictly convex. Evidently, \(\hat{V} - c_{\varepsilon}\) is strictly concave for all \(\varepsilon > 0\), so \(\delta_{\mu_0}\) is the uniquely optimal solution to the information acquisition problem in \(\hat{\mathcal{D}}\). Now consider
    \[\mathbb{E}_{\Phi^{\dagger}}\left[V - c_{\varepsilon}\right] = \mathbb{E}_{\Phi^{\dagger}}\left[V - \hat{V} - \varepsilon \rho\right] = -\mathbb{E}_{\Phi^{\dagger}}\hat{W} - \varepsilon \mathbb{E}_{\Phi^{\dagger}} \rho\text{.}\]
    This expression is continuous in \(\varepsilon\) and is strictly greater than \(-\hat{W}\left(\mu_0\right)\) when \(\varepsilon = 0\). Accordingly, for all sufficiently small \(\varepsilon > 0\), this expression is strictly greater than \(-\hat{W}\left(\mu_0\right)\).\end{proof}

    \subsection{Corollary \ref{incremental} Proof}\label{incrementalproof}
    \begin{proof}
    \(\left(\Leftarrow\right)\) Lemma \ref{lemmany} implies this direction.

    \medskip

    \(\left(\Rightarrow\right)\) This is an easy implication of Lemma \ref{newexoginfo}: just let \(\Upsilon\) be a binary Bayes-plausible distribution for which one point of support is the \(\mu_0^{*}\) specified in the proof of Lemma \ref{newexoginfo} and the other point of support is some \(\mu'' \in \Delta\left(\Theta\right)\) and \(\Phi\) is the mean-preserving spread of \(\Upsilon\) supported on \(\left\{\mu,\mu',\mu''\right\}\), where \(\mu\) and \(\mu'\) are the posteriors specified in the proof of Lemma \ref{newexoginfo}. Intuitively, \(\Phi\) is obtained from \(\Upsilon\) by first observing the realized posterior from \(\Upsilon\) then learning further if the realization is \(\mu_{0}^*\). Evidently,
    \[\mathbb{E}_{\Phi}\hat{V}\left(\mu\right) - \mathbb{E}_{\Upsilon}\hat{V}\left(\mu\right) < \mathbb{E}_{\Phi}V\left(\mu\right) - \mathbb{E}_{\Upsilon}\textcolor{OrangeRed}{V}\left(\mu\right)\text{,}\]
    by construction.\end{proof}

        \subsection{Proposition \ref{prioroptimal} Proof}\label{prioroptimalproof}
    \begin{proof}
    Making the agent more flexible implies that \(\hat{V}\) lies weakly above \(V\) pointwise. Likewise, as some action remains prior-optimal \(V\left(\mu_{0}\right) = \hat{V}_{\mu_0}\). Therefore, Theorem \ref{mugenerate} produces the first statement (\ref{statement1}). To prove the the second and third statements  (\ref{statement2} and \ref{statement3}), suppose for the sake of contradiction that some action remains prior-optimal. As the agent is less flexible in \(\hat{D}\) and the removal is consequential, \(\hat{V}\) lies weakly below \(V\) pointwise, and there exists \(\mu' \in \Delta(\Theta)\) at which \(\hat{V}(\mu') < V(\mu')\). Thus, \(\hat{V}\) does not shift-majorize \(V\), and so Theorems \ref{mugenerate} and \ref{muless} allow us to conclude the results.\end{proof}

    \subsection{Proposition \ref{badremoval} Proof}\label{badremovalproof}

    \begin{proof}
    We shall prove the contrapositive statement. Let \(\hat{V} \neq V\) and \(\hat{a}_{\mu_{0}}\) not be weakly dominated by \(a_{\mu_{0}}\) in \(\mathcal{D}\). If \(\hat{a}_{\mu_{0}} = a_{\mu_{0}}\), Proposition \ref{prioroptimal} implies \(\hat{V}\) does not shift-majorize \(V\).
    
    Now suppose \(\hat{a}_{\mu_{0}} \neq a_{\mu_0}\). Let \(\alpha \cdot \mu + \beta\) and \(\hat{\alpha} \cdot \mu + \hat{\beta}\) (\(\alpha, \hat{\alpha} \in \mathbb{R}^{n-1}\) and \(\beta, \hat{\beta} \in \mathbb{R}\)) denote the expected payoffs to actions \(a_{\mu_{0}}\) and \(\hat{a}_{\mu_{0}}\). We want to show that \(\hat{V}\) does not shift-majorize \(V\). Evidently, the generic prior means that the only candidate affine modifier is \(\ell\left(\mu\right) = \left(\alpha -\hat{\alpha}\right)  \cdot \mu + \beta- \hat{\beta}\). 
    
    Suppose for the sake of contradiction that \(\hat{V} + \ell\) dominates \(V\) pointwise on \(\Delta\left(\Theta\right)\); \textit{viz.}, for all \(\mu \in \Delta\left(\Theta\right)\),
    \[\hat{V}\left(\mu\right) + \left(\alpha -\hat{\alpha}\right)  \cdot \mu + \beta- \hat{\beta} \geq V\left(\mu\right)\text{.}\]
    Observe that we can write
    \[\hat{V}\left(\mu\right) = \max\left\{T\left(\mu\right), \hat{\alpha} \cdot \mu + \hat{\beta}\right\} \quad \text{and} \quad V\left(\mu\right) = \max\left\{S\left(\mu\right), T\left(\mu\right), \alpha \cdot \mu + \beta, \hat{\alpha} \cdot \mu + \hat{\beta}\right\}\text{,}\]
    where \(T\) and \(S\) are convex, piecewise-affine functions (\(S\) is generated by the actions other than \(a_{\mu_0}\) that are removed from the decision problem).

    Therefore, for all \(\mu \in \Delta\left(\Theta\right)\),
    \[\begin{split}
        \max\left\{T\left(\mu\right), \hat{\alpha} \cdot \mu + \hat{\beta}\right\} + \left(\alpha -\hat{\alpha}\right)  \cdot \mu + \beta- \hat{\beta} &\geq \max\left\{S\left(\mu\right),T\left(\mu\right), \alpha \cdot \mu + \beta, \hat{\alpha} \cdot \mu + \hat{\beta}\right\}\\
         \max\left\{T\left(\mu\right), \hat{\alpha} \cdot \mu + \hat{\beta}\right\} + \left(\alpha -\hat{\alpha}\right)  \cdot \mu + \beta- \hat{\beta} &\geq \max\left\{T\left(\mu\right), \alpha \cdot \mu + \beta, \hat{\alpha} \cdot \mu + \hat{\beta}\right\}\\
         \left(\alpha -\hat{\alpha}\right)  \cdot \mu + \beta- \hat{\beta} &\geq \max\left\{\alpha \cdot \mu + \beta - \max\left\{T\left(\mu\right), \hat{\alpha} \cdot \mu + \hat{\beta}\right\}, 0\right\}\\ &\geq 0
        \text{,}
    \end{split}\]
so \(a_{\mu_{0}}\) dominates \(\hat{a}_{\mu_{0}}\) in \(\mathcal{D}\), a contradiction. \end{proof}

\end{document}